\documentclass[11pt,letterpaper]{article}
\pdfoutput=1

\usepackage[T1]{fontenc}
\usepackage{graphicx}
\usepackage{amsmath,amsthm,amssymb,amsfonts}
\usepackage[left=1in,right=1in,top=1in,bottom=1in]{geometry}
\usepackage{comment}
\usepackage{todonotes}
\usepackage{tikz}
\usepackage{microtype}
\usepackage{subcaption} 
\usepackage[noadjust]{cite}
\usepackage{mathtools}
\usepackage{hyperref}
\usepackage{enumitem}
\usepackage[capitalize]{cleveref}
\usepackage{wrapfig,lipsum,booktabs}
\usepackage{authblk}

\title{Optimal Bounds for Distinct Quartics}

\newcommand{\cO}{\mathcal{O}}
  \newcommand{\floor}[1]{\left\lfloor #1 \right\rfloor}

  \newcommand{\can}[2]{C_{(#1,#2)}}
\newcommand{\width}{\textsf{width}}
\newcommand{\height}{\textsf{height}}
\newcommand{\BR}{\textsf{BR}}

\def\dd{\mathinner{.\,.}}
\newcommand{\OCR}{\textsf{occ}_{5\times 5}}
\newcommand{\assign}{\textsf{assign}}
\newcommand{\hspan}[2]{\textsf{row}_{3\cdot \height(#1)}(#2^{2,2})}
\newcommand{\vspan}[2]{\textsf{col}_{3\cdot \width(#1)}(#2^{2,2})}
\renewcommand{\P}{\mathcal{P}}
\renewcommand{\S}{\mathcal{S}}
\newcommand{\C}{\mathcal{C}}
\newcommand{\R}{\mathcal{R}_{\S}}

\newcommand{\hmeta}[2]{H_{#1}^{#2}}
\newcommand{\vmeta}[2]{V_{#1}^{#2}}
\newcommand{\per}{\textsf{per}}
\newcommand{\IPM}{\textsf{IPM}}

\newtheorem{theorem}{Theorem}[section]
\newtheorem{proposition}[theorem]{Proposition}
\newtheorem{observation}[theorem]{Observation}
\newtheorem{lemma}[theorem]{Lemma}
\newtheorem{corollary}[theorem]{Corollary}

\newtheorem{fact}[theorem]{Fact}

\author[1]{Panagiotis Charalampopoulos}
\author[2]{Paweł Gawrychowski}
\author[3,4]{Samah Ghazawi}

\affil[1]{School of Computing and Mathematical Sciences, Birkbeck, University of London, UK\\
    \texttt{p.charalampopoulos@bbk.ac.uk}}

\affil[2]{Institute of Computer Science, University of Wrocław, Poland\\
    \texttt{gawry@cs.uni.wroc.pl}}

\affil[3]{Department of Software Engineering, Braude, College of Engineering, Karmiel, Israel}
    
\affil[4]{Department of Computer Science, University of Haifa, Israel\\
    \texttt{samahi@braude.ac.il}}

\makeatletter

\begin{document}

\date{\vspace{-1cm}}
\maketitle
\thispagestyle{empty}

\begin{abstract}
A fundamental concept related to strings is that of repetitions. It has been extensively studied in many versions, from both purely combinatorial and algorithmic angles. One of the most basic questions is how many distinct squares, i.e., distinct strings of the form $UU$, a string of length~$n$ can contain as fragments. It turns out that this is always $\cO(n)$, and the bound cannot be improved to sublinear in $n$ [Fraenkel and Simpson, JCTA 1998].

Several similar questions about repetitions in strings have been considered, and by now we seem to have a good understanding of their repetitive structure. For higher-dimensional strings, the basic concept of periodicity has been successfully extended and applied to design efficient algorithms---it is inherently more complex than for regular strings. Extending the notion of repetitions and understanding the repetitive structure of higher-dimensional strings is however far from complete.

Quartics were introduced by Apostolico and Brimkov [TCS 2000] as analogues of squares in two dimensions. Charalampopoulos, Radoszewski, Rytter, Waleń, and Zuba [ESA 2020] proved that the number of distinct quartics in an $n\times n$ 2D string is $\cO(n^2 \log^2 n)$ and that they can be computed in $\cO(n^2 \log^2 n)$ time. Gawrychowski, Ghazawi, and Landau [SPIRE 2021] constructed an infinite family of $n \times n$ 2D strings with $\Omega(n^2 \log n)$ distinct quartics. This brings the challenge of determining asymptotically tight bounds. Here, we settle both the combinatorial and the algorithmic aspects of this question: the number of distinct quartics in an $n\times n$ 2D string is $\cO(n^2 \log n)$ and they can be computed in the worst-case optimal $\cO(n^2 \log n)$ time.

As expected, our solution heavily exploits the periodic structure implied by occurrences of quartics.
However, the two-dimensional nature of the problem introduces some technical challenges.
Somewhat surprisingly, we overcome the final challenge for the combinatorial bound using a result of Marcus and Tardos [JCTA 2004] for permutation avoidance on matrices.
\end{abstract}

\clearpage
\setcounter{page}{1}
\section{Introduction}
Repetitions are a staple topic of both combinatorics on words~\cite{BerstelP07} and algorithms on strings~\cite{AlgStrings}.
In both areas, the classical objects of study are linear sequences of characters from a finite alphabet.
Depending on whether we are more interested in their combinatorial properties or designing efficient algorithms for them,
it is customary to call such sequences words or strings, respectively. In this paper, we use the latter convention.

Perhaps the most natural example of a repetition in a string is a square, that is, a string of the form $UU$, also known as a ``tandem repeat'' in the biological literature~\cite{Gusfield1997}.
The basic question concerning squares is whether any of the fragments of a string of length $n$ is a square,
and, if so, what is the number of such fragments.
The origins of this question can be traced back to Thue~\cite{thue1906}, who
constructed an infinite string over a ternary alphabet that contains no squares.
Thus, we can construct arbitrarily long square-free strings over such alphabets.
The next question is what is the largest possible number of fragments that are squares.
However, any even-length fragment of $\texttt{a}^{n}$ is a square.
One way to make the question non-trivial is to only consider the primitively rooted squares,
meaning that $U$ is not a power of another string. This decreases the possible number of
occurrences to $\cO(n\log n)$, which is asymptotically tight~\cite{DBLP:journals/ipl/Crochemore81}.
Another way is to only consider distinct squares.

Fraenkel and Simpson~\cite{DBLP:journals/jct/FraenkelS98} showed that any string of length $n$ contains
at most $2n$ distinct squares and constructed an infinite family of strings of length $n$ containing
$n-\Theta(n)$ distinct squares. For many years, it was conjectured that the upper bound should be at most $n$.
After a series of simplifications and improvements~\cite{Ilie05,Ilie07,Lam,DEZA,Thierry2020},
the conjecture was finally proven by Brlek and Li~\cite{Brlek}, who showed an upper bound of $n-\sigma+1$,
where $\sigma$ is the size of the alphabet.
The same authors~\cite{DBLP:conf/cwords/BrlekL23} also showed an upper bound of $n-\Theta(\log n)$.
On the algorithmic side, Apostolico and Preparata~\cite{ApostolicoP83},
Main and Lorentz~\cite{Main1984} and Crochemore~\cite{Crochemore1981}
showed how to find a compact representation of all squares (in particular, test square-freeness)
in a string of length $n$ in $\cO(n\log n)$ time. Specifically, such a representation stores
all distinct squares. To obtain a faster algorithm for finding only the distinct squares, one needs
to restrict the size of the alphabet.
For constant alphabets, Gusfield and Stoye~\cite{GusfieldS04} designed an $\cO(n)$ time algorithm.
This was later generalized to the more general case of an integer alphabet (that can be sorted in linear
time)~\cite{BannaiIK17,CrochemoreIKRRW14}. The complexity of testing square-freeness over general ordered and unordered alphabets was very recently settled by Ellert and Fischer~\cite{ellert_et_al:LIPIcs.ICALP.2021.63} and Ellert et al.~\cite{DBLP:conf/soda/EllertGG23}, respectively; this problem has been also studied in the
parallel~\cite{Apostolico92,Apostolico1996,CrochemoreR91,CrochemoreR91a} and
the online settings~\cite{LeungPT06,Hong2008,Kosolobov2014,Kosolobov2015a}.
Thus, by now we seem to have obtained a rather good understanding
of both the combinatorial and the algorithmic properties of distinct squares.

Arguably, linear sequences are not always best suited to model the objects that we would like to study.
A natural extension is to consider rectangular arrays of characters from a finite alphabet, which can be seen as 2D strings.
Possible applications in image processing~\cite{RosenfeldKak82} sparked interest in designing algorithms
for searching in 2D strings already in the late 1970s~\cite{Baker78a,Bird77}.
This turned out to be significantly more challenging than searching in 1D strings: both versions were
studied already in the 70s, but while for 1D strings an alphabet-independent linear-time algorithm
had been soon found~\cite{DBLP:journals/siamcomp/KnuthMP77}, achieving the same goal for
2D strings took till the 90s~\cite{DBLP:journals/siamcomp/AmirBF94,GalilP96,DBLP:conf/latin/CrochemoreR95}.
Extensions of this basic problem such as approximate searching~\cite{DBLP:journals/iandc/AmirF95,DBLP:conf/spire/CliffordFSV16}, indexing~\cite{DBLP:journals/siamcomp/Giancarlo95,DBLP:conf/icalp/GiancarloG95,DBLP:journals/ipl/ChoiL97}, searching in smaller space~\cite{CrochemoreGPR95},
scaled searching~\cite{DBLP:journals/algorithmica/AmirBLP09,DBLP:journals/algorithmica/AmirC10},
searching in random 2D strings~\cite{DBLP:journals/siamcomp/KarkkainenU99},
dictionary searching~\cite{DBLP:journals/ipl/AmirF92,DBLP:journals/algorithmica/NeuburgerS13,DBLP:journals/iandc/IduryS95}, and
searching in compressed 2D strings~\cite{DBLP:journals/jal/AmirBF97,DBLP:journals/jal/AmirLS03,DBLP:conf/dcc/AmirB92}
have been also considered.

The combinatorial structure of 2D strings seems to be significantly more involved than that of 1D strings.
As a prime example, the basic tool used in algorithms and combinatorics on 1D strings is periodicity.
We say that $p$ is a period of a 1D string $S[1\dd n]$ when $S[i]=S[i+p]$ for all $i=1,2,\ldots,n-p$.
The set of all periods is very structured due to a classical result of Fine and Wilf~\cite{fine1965uniqueness}
according to which for any two periods $p,q$ such that $p+q\leq n$, the greatest common divisor of
$p$ and $q$ is also a period. The natural way to extend this notion to 2D strings is to define
$(x,y)$ to be a period of a 2D string $S[1 \dd n][1 \dd n]$ when $S[i+x][j+y]$ for all $i=1,2,\ldots,n-x$ and $j=1,2,\ldots,n-y$.
This notion was introduced by Amir and Benson~\cite{AmirB98}, who provided a detailed study
based on classifying 2D strings into four periodicity classes.
This classification was later crucial in designing both an alphabet-independent linear-time
algorithm~\cite{GalilP96,DBLP:conf/latin/CrochemoreR95} and alphabet-independent optimal parallel
algorithms~\cite{DBLP:journals/siamcomp/CrochmoreGHMR98,DBLP:journals/iandc/AmirBF98}.

The rich combinatorial structure of 2D strings brings the challenge of finding the right generalization
of the concept of repetitions.
Apostolico and Brimkov~\cite{AB} introduced two notions of repetitions in 2D strings that can be seen
as natural analogues of squares in 1D strings. A tandem $W^{1,2}$ (or $W^{2,1}$) consists of 2 occurrences of the same block
$W$ arranged in a $1\times 2$ (or $2\times 1$) pattern. Next, a quartic $W^{2,2}$ consists of 4 occurrences of the
same block $W$ arranged in a $2\times 2$ pattern. Note that Apostolico and Brimkov~\cite{AB} additionally
required that $W$ is primitive, meaning that it cannot be partitioned into non-overlapping copies
of another block. However, it is more natural to call such tandems and quartics primitively rooted,
as in~\cite{CRRWZ}.
Apostolico and Brimkov~\cite{AB} showed asymptotically tight bounds of $\cO(n^{3}\log n)$ and $\cO(n^{2}\log^{2}n)$ for the number of primitively rooted tandems and quartics, respectively.
The former bound was later complemented with a worst-case optimal $\cO(n^{3}\log n)$-time
algorithm~\cite{DBLP:journals/dam/ApostolicoB05}.

Two tandems $T=W^{1,2}$ and $T'=V^{1,2}$ are distinct when $W\neq V$. 
Similarly,
two quartics $Q=W^{2,2}$ and $Q'=V^{2,2}$ are distinct when $W\neq V$.
It is easy to see that an $n\times n$ 2D string contains $\cO(n^{3})$ distinct tandems by applying the bound on the number of 1D distinct squares on every horizontal slice of the 2D string. It is also not hard to show that this bound is asymptotically
tight, even over a binary alphabet~\cite{GGL}. Thus, tandems do not seem to be the right generalization of squares,
and we should rather focus on quartics.

Recently, Charalampopoulos, Radoszewski, Rytter, Waleń, and Zuba~\cite{CRRWZ}
showed a non-trivial upper bound of $\cO(n^{2}\log^{2}n)$ on the number of distinct quartics in an $n\times n$ 2D string,
and an algorithm that finds them in the same time complexity.
At this point, it was quite unclear to what extent distinct quartics suffer from the
``curse of dimensionality''. Could it be that, similarly to the number of distinct squares, their number is also linear in the size of the
input? Gawrychowski, Ghazawi, and Landau~\cite{GGL} very recently showed that this is not the case,
by constructing an infinite family of $n\times n$ 2D strings over a binary alphabet containing $\Omega(n^{2}\log n)$
distinct quartics. This shows that there is a qualitative difference between distinct squares and distinct quartics,
but leaves a significant gap between the lower bound of $\Omega(n^{2}\log n)$ and the upper bound of
$\cO(n^{2}\log^{2}n)$~\cite{CRRWZ}.

\paragraph{Our Results.}
Our contribution is twofold. First, we show an asymptotically tight bound of $\cO(n^{2}\log n)$ on the number of
distinct quartics in an $n\times n$ 2D string.
Thus, the ``curse of dimensionality'' for this problem is a single logarithm for going from 1D to 2D.
Second, we show how to find all distinct quartics in worst-case optimal $\cO(n^{2}\log n)$ time.
We thus resolve both the combinatorial and algorithmic complexity of distinct quartics.

A notable difference of our algorithm from the previously fastest algorithm for computing distinct quartics~\cite{CRRWZ} is that the algorithm of~\cite{CRRWZ} first finds all 2D runs\footnote{2D runs are subarrays that are periodic both vertically and horizontally and cannot be extended without any of the periods changing.} of the 2D string, which are not even known to be $\cO(n^2 \log n)$, and then infers the quartics from those.
We manage to circumvent this, by focusing on some selected occurrences of 2D strings of the form $Q^{5,5}$, instead of considering all of them via 2D runs.

\paragraph{Overview of the Combinatorial Upper Bound.}
When bounding the number of distinct squares, one begins with fixing the rightmost occurrence of every distinct
square~\cite{DBLP:journals/jct/FraenkelS98}. In two dimensions, it is less clear what an extreme occurrence could mean.
We simply say that it is an occurrence at a position $(i,j)$ such that there is no other occurrence at a different position $(i',j')$
such that $i' \geq i$ and $j'\geq j$.
Next, a standard trick used when working with strings is to partition them into groups with length in $[2^{a} \dd 2^{a+1})$
for different integers $a$.
Similarly to previous work~\cite{CRRWZ}, we partition quartics into groups $\can{a}{b}$ with height in $[2^{a} \dd 2^{a+1})$
and width in $[2^{b} \dd 2^{b+1})$ for pairs of integers $(a,b)$.
We begin with proving that, for any position $(i,j)$, the set of extreme occurrences at $(i,j)$ may have a non-empty
intersection with only $\cO(\log n)$ such groups.
Next, we partition all quartics into \emph{thin} and \emph{thick} (note that the meaning of thin and thick is slightly different
than in the previous work~\cite{CRRWZ}).
More specifically, a quartic $Q$ is thick if and only if it can be partitioned into $x \times y$ occurrences of a primitive 2D string $R$, i.e., $Q = R^{x,y}$ for some $x,y\geq 5$.
Then, we show that for any position $(i,j)$ and group $\can{a}{b}$, there can be at most 10 extreme occurrences
of thin quartics in $\can{a}{b}$ at position $(i,j)$.
Overall, we thus have only $\cO(n^{2}\log n)$ distinct thin quartics.

The main part of our proof for the combinatorial upper bound is the analysis of the number of extreme occurrences of thick quartics.
Our starting point is the observation (already present in~\cite{CRRWZ}) that this number can be upper bounded by the number of occurrences of 2D strings of the form $R^{5,5}$, for primitive $R$,
that participate in the partition of an extreme occurrence of some quartic $R^{x,y}$.
To bound the number of such occurrences, we assign an occurrence of $R^{5,5}$ at position $(i,j)$
to position $(x,y)=(i+2\cdot\height(R),j+2\cdot \width(R))$ and say that this occurrence is \emph{anchored} at position $(x,y)$. Then, our goal is to show that the number of occurrences
assigned to every position is only $\cO(\log n)$. For a fixed position $(i,j)$, this is done by first arguing that
the pairs $(\floor{\log (\height(R))}, \floor{\log (\width(R))})$ are pairwise distinct among occurrences
of different $R^{5,5}$ assigned to $(i,j)$. This requires a careful analysis of the implied periodic structure and allows us to focus on bounding the number of such pairs.
What we do next is the main novelty of our approach for the combinatorial upper bound.
We treat the pairs as a set of points $\P \subseteq [1\dd m]^{2}$, where $m=\floor{\log n}$, and argue that, 
for each $(a,b)\in \P$, the set of points of $\P$ that are strictly dominated by $(a,b)$ can be partitioned
into at most two chains. 
Next, our goal is to upper bound the size of any set $\P$ with this property by $\cO(m)$.
To this end, we leverage a result from extremal combinatorics, namely, the proof
of the Füredi-Hajnal conjecture by Marcus and Tardos~\cite{DBLP:journals/jct/MarcusT04}.
This result states that, if an $m\times m$ binary matrix $M$ avoids a fixed permutation matrix $P$ as a submatrix, i.e., if $P$ cannot be obtained by deleting some rows and columns of $M$ and changing 1s to 0s,
then it contains at most $c_P \cdot m$ 1s, where $c_P$ is a constant if the size of $P$ is a constant. We reformulate the constraint on~$\P$ to avoid
the permutation matrix shown below as a submatrix.
Overall, this allows us to conclude that the number
of extreme occurrences of thick quartics is also $\cO(n^{2}\log n)$.

\begin{center}
\begin{tabular}{ |c|c|c|c| } 
 \hline
 & &  1 & \\ 
 \hline
 & 1 &  & \\ 
 \hline
 1 & &  & \\ 
 \hline
 & &  & 1 \\ 
 \hline
\end{tabular}
\end{center}

A high-level description of our approach for the algorithmic part is provided in~\cref{sec:over}, as it is best read after the full proof of the combinatorial upper bound.

\paragraph{Open Problem.}
An interesting follow-up question on repetitions in 2D strings is that of settling the number of 2D runs that a 2D string can
have.
Charalampopoulos et al.~\cite{CRRWZ} proved an $\cO(n^2\log^2 n)$ upper bound for the number of 2D runs that an
$n \times n$ 2D string can contain, while Gawrychowski et al.~\cite{GGL} constructed an infinite family of $n\times n$ 2D strings
(over a binary alphabet) with $\Omega(n^2 \log n)$ 2D runs.
On the algorithms' side, Amir et al.~\cite{Amir2020} devised an algorithm that computes all 2D runs in an $n \times n$ 2D
string in $\cO(n^2 \log n + |\textsf{output}|)$ time, and is thus optimal.
For 1D strings, after a long line of results~\cite{KK:99,Rytter2006,Crochemore2008,Crochemore2011,Giraud2008,Giraud2009,PuglisiSS08} the number of runs was shown to be less than $n$~\cite{runstheorem} and they can be computed in $\cO(n)$ time for strings over ordered alphabets~\cite{ellert_et_al:LIPIcs.ICALP.2021.63} (see~\cite{KK:99,runstheorem} for earlier algorithms for strings over linear-time sortable alphabets).

\section{Preliminaries}
\label{preliminaries}
For integers $i\leq j$,
we denote the integer interval $\{i, \ldots, j\}$ by either of $[i \dd j]$, $(i-1 \dd j+1)$, $[i \dd j+1)$, and $(i-1 \dd j]$.

Let us consider a string $S=S[1]S[2]\cdots S[n]$ of length $|S|=n$.
For integers $i \leq j$ in $[1 \dd n]$, we denote the \emph{fragment} $S[i]\cdots S[j]$ by $S[i \dd j]$.
A non-trivial rotation of a string $S$ of length $n$ is a string $S[j\dd n]S[1\dd j]$ for $j \in [1\dd n)$.
A positive integer $p \leq n$ is a \emph{period} of $S$  if and only if $S[i]=S[i+p]$ for all $i \in [1 \dd n-p]$.
The smallest period of $S$ is called \emph{the period} of $S$ and is denoted by $\per(S)$.
A string is called \emph{periodic} if and only if its period is at most half its length.
We will extensively use the following property of periods.

\begin{lemma}[Periodicity Lemma~\cite{fine1965uniqueness}]\label{lem:FW}
If $p$ and $q$ are periods of a string $S$ and satisfy $p + q \leq |S|$, then $\gcd(p, q)$ is also a period of $S$.
\end{lemma}

We denote the concatenation of two strings $U$ and $V$ by $UV$.
Further, for $k \in \mathbb{Z}_+$, we denote the concatenation of $k$ copies of $U$ by $U^k$.
A string $V$ that cannot be written as $U^k$ for a string $U$ and an integer $k>1$ is called \emph{primitive}.
A string of the form $UU$ is called a \emph{square}.
A square $UU$ is said to be primitively rooted if $U$ is primitive.
More generally, a string of the form $U^{k}$ is called a $k$-th power, and it is said to be primitively
rooted if $U$ is primitive.
We extensively use the following property of squares.

\begin{lemma}[Three Squares Lemma~\cite{DBLP:journals/algorithmica/CrochemoreR95}\protect\footnotemark]\label{lem:three_sq}
\footnotetext{This formulation comes from~\cite{DBLP:journals/jct/FraenkelS98}.}
If squares $U^2$ and $V^2$ are proper prefixes of a square $W^2$, $|U|<|V|$, and $U$ is primitive, then $|U|+|V|\leq |W|$.
\end{lemma}

We next summarise some combinatorial properties of squares and higher powers.

\begin{proposition}\label{lem:two}
Consider a string $S$ and an integer $a$. At most two prefixes of $S$ with lengths from $[2^{a} \dd 2^{a+1})$ can be primitively
rooted squares.
\end{proposition}

\begin{proof}
Assume that there are three such prefixes, and denote them by $UU$, $VV$, $WW$, where
$|U|<|V|<|W|$. Since $|UU|, |VV|, |WW|\in [2^{a} \dd 2^{a+1})$, we have $|U|+|V|>|W|$, which together with the primitivity of $|U|$
leads to a contradiction. 
\end{proof}

\begin{proposition}\label{lem:three}
Consider a string $S$ and an integer $a$. All prefixes of $S$ with lengths from $[2^{a} \dd 2^{a+1})$ that are
powers higher than 2 are of the form $U^{k}$ for the same primitive string $U$.
\end{proposition}

\begin{proof}
Assume that there are two such prefixes $U^{k}$ and $V^{\ell}$, where $U$ is primitive, $k,\ell \geq 3$
and $|U| < |V|$. Since $|U^{k}|, |V^{\ell}|\in [2^{a} \dd 2^{a+1})$ and $k,\ell \geq 3$, we have $|U|+|V| \leq |U^{k}|$.
Next, $|U|$ and $|V|$ are periods of $U^{k}$, so by Lemma~\ref{lem:FW} we obtain that $\gcd(|U|,|V|)$ is a period
of $U^{k}$. But $U$ is primitive, so $\gcd(|U|,|V|)=|U|$, and $V$ is a power of $U$.
By repeating this reasoning, we obtain that all such prefixes are powers of the same primitive $U$.
\end{proof}

A fragment $S[i \dd j]$ of a string $S$ is a \emph{run} if and only if it is periodic and it cannot be extended by a character in either direction with its period remaining unchanged.

An $m \times n$ 2D string $A$ is simply a two-dimensional array with $m$ rows and $n$ columns, where $\height(A)=m$ and $\width(A)=n$.
The position that lies on the $i$-th row and the $j$-th column of $A$ is position $(i,j)$. 
We say that a 2D string $P$ occurs at a position $(i,j)$ of a 2D string~$T$ if and only if the \emph{subarray} (also called \emph{fragment}) $T[i \dd i+\height(P))[j \dd j+\width(P))$ of $T$ equals $P$.
We write $\Sigma^{*,*}$ to denote the set of all 2D strings over alphabet $\Sigma$.

A positive integer $p$ is a \emph{horizontal period} of a 2D string $A$ such that $\width(A) \geq p$ if and only if the $j$-th column of $A$ is equal to the $(j+p)$-th column of $A$ for all $j \in [1 \dd \width(A)-p]$.
The smallest horizontal period of $A$ is \emph{the horizontal period} of $A$.
An integer $q$ is a (the) \emph{vertical period} of $A$ if and only if $q$ is a (resp.~the) horizontal period of the transpose of~$A$.

It will be sometimes convenient to view a 2D string as a 1D metastring by viewing each column (or row) as a metacharacter such that metacharacters are equal if and only if the corresponding columns (resp.~rows) are equal. Observe that the horizontal periods (resp.~vertical periods) of a 2D string $A$ are in one-to-one correspondence with the periods of the metastring obtained from $A$ by viewing each column (resp.~row) as a metacharacter.

For a 2D string $W$ and $x, y \in \mathbb{Z}_+$, we denote by $W^{x,y}$ the 2D string that consists of $x\times y$ copies of $W$; see Figure~\ref{fig:2x5w} for an illustration.
A 2D string $W$ is \emph{primitive} if it cannot be written as $Y^{a,b}$ for any 2D string $Y$ and $a,b \in \mathbb{Z}_+$ that are not both equal to $1$.
The primitive root of a 2D string~$X$ is the unique primitive 2D string $Y$ such that $X=Y^{a,b}$ for $a,b \in \mathbb{Z}_+$.
Note that the primitive root is indeed unique by the periodicity lemma applied to the horizontal and vertical 1D metastrings obtained from $X$.

\begin{figure}[t]
\centering
\begin{tabular}{|c|c|c|c|c|}
    \hline
   W & W &W & W&W\\
   \hline
   W & W &W & W&W\\
    \hline
\end{tabular}
\caption{2D string $W^{2,5}$ is shown for some 2D string $W$.}\label{fig:2x5w}
\end{figure}

\subparagraph{Model of computation.}
For our algorithm, we assume the standard word-RAM model of computation with word-size $\Omega(\log n)$.
\section{The Combinatorial Bound}

We consider an $n \times n$ 2D string $A$, whose entries are over an arbitrary alphabet $\Sigma$.
We say that a fragment $A[i \dd i')[j \dd j')$ is a quartic-fragment if and only if it equals some quartic $Q$; further, we say that it is an \emph{extreme} or \emph{bottom-right} quartic-fragment if $Q$ does not have any occurrence at another position $(i'',j'')$ with $i'' \geq i$ and $j'' \geq j$.
We refer to such an occurrence of $Q$ as an extreme or bottom-right occurrence.
We denote by $\BR(i,j)$ the set of extreme quartic-fragments with top-left corner $(i,j)$.
Further, we denote the union of all $\BR(i,j)$ by $\BR$.
Observe that the distinct quartics in $\BR$ are exactly the distinct quartics in $A$ as every quartic that occurs in $A$ has at least one extreme occurrence.
Note that a quartic may have $\Theta(n)$ extreme occurrences; an example is provided in Figure~\ref{fig:extreme}.

\begin{figure}[ht]
\centerline{\includegraphics[scale=.8]{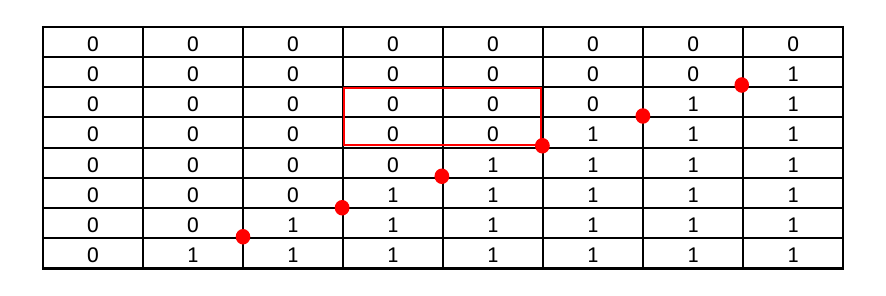}}
\caption{Consider an $n \times n$ 2D string $A$ all of whose entries that lie weakly above the main diagonal are equal to $0$ and all of whose entries that lie strictly below the main diagonal are equal to~$1$. The quartic that equals $0^{2,2}$ has $n-2$ extreme occurrences in $A$.
This is illustrated for $n=8$: the bottom-right corners of extreme occurrences of said quartic are marked.}
\label{fig:extreme}
\end{figure}

Let us consider a partition of the quartic-fragments of $A$ into $\cO(\log^2 n)$ \emph{canonical sets}, such that, for each $(a,b) \in [1 \dd \floor{\log n}]^2$,
the canonical set $\can{a}{b}$ consists of all quartic-fragments of $A$
whose height is in $[2^a \dd 2^{a+1})$ and whose width is in $[2^b \dd 2^{b+1})$.\footnote{Throughout this work, logarithms have base $2$.}

\begin{lemma}\label{lem:aspect_ratio}
For each position $(i,j)$ of $A$, $\BR(i,j)$ has a non-empty intersection with $\cO(\log n)$ canonical sets.
\end{lemma}
\begin{proof}
We say that the \emph{aspect ratio} of a quartic $Q$ is equal to $2$ raised to the power $\floor{\log(\height(Q))} - \floor{\log(\width(Q))}$.
The aspect ratio of all quartic-fragments in a canonical set $\can{a}{b}$ is $2^{a-b}$.
Observe, that there are $2\cdot\floor{\log n}-1$ different possible values for the aspect ratio of a quartic.
For each $d \in [-\floor{\log n}+1 \dd \floor{\log n}-1]$, let $\BR_d(i,j)$ be the subset of $\BR(i,j)$ that contains exactly the elements of $\BR(i,j)$ with aspect ratio $2^d$.

Next, we show that, for each $d$, we have at most two canonical sets contributing to $\BR_{d}(i,j)$.
Let us suppose towards a contradiction that we have three canonical sets $\can{a}{b}$, $\can{a'}{b'}$, and $\can{a''}{b''}$ that contribute to $\BR_{d}(i,j)$.
In other words, there are quartic-fragments
\begin{itemize}
\item $Q\in\BR_{d}(i,j)$ with height in $[2^a \dd 2^{a+1})$ and width in $[2^b \dd 2^{b+1})$,
\item $Q'\in\BR_{d}(i,j)$ with height in $[2^{a'} \dd 2^{a'+1})$ and width in $[2^{b'} \dd 2^{b'+1})$, and
\item $Q''\in\BR_{d}(i,j)$ with height in $[2^{a''} \dd 2^{a''+1})$ and width in $[2^{b''} \dd 2^{b''+1})$.
\end{itemize}
Since $a-b=a'-b'=a''-b''=d$, we can assume without loss of generality that $a<a'<a''$ and $b<b'<b''$.
We thus have that $a+1<a''$ and $b+1<b''$, which implies that $Q$ is fully contained in the top left quarter of $Q''$.
Thus, $Q$ has an occurrence at position $(i+\height(Q'')/2,j+\width(Q'')/2)$; see Figure~\ref{fig:aspect_ratio}.
This contradicts our assumption that the occurrence of $Q$ at position $(i,j)$ is an extreme occurrence.

Thus, $\cO(\log n)$ canonical sets contribute to $\BR(i,j)$: at most two for each aspect ratio.
\end{proof}
\begin{figure}[ht]
\centerline{\includegraphics[scale=.5]{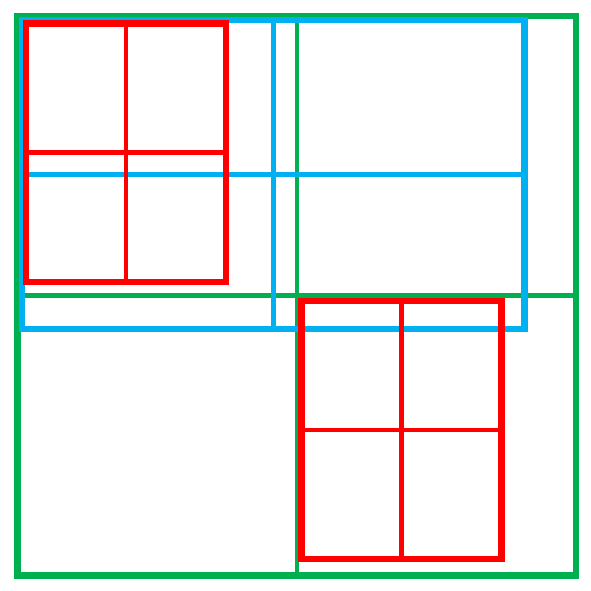}}
\caption{An illustration of the proof of \cref{lem:aspect_ratio} with
quartics $Q$, $Q'$, and $Q''$ drawn in red, blue, and green, respectively.}
\label{fig:aspect_ratio}
\end{figure}

Henceforth, we call a quartic $Q$ with primitive root $P$ \emph{thick} if $Q = P^{x,y}$ for $x,y \geq 5$ and \emph{thin} otherwise.

\begin{lemma}\label{lem:thick_period}
For any position $(i,j)$ of $A$ and any pair $(a,b) \in [1 \dd \floor{\log n}]^2$, $\can{a}{b} \cap \BR(i,j)$ can contain at most 10 thin quartics.
\end{lemma}
\begin{proof}
The possible forms of \emph{thin} quartics are
$P^{2,2}$,
$P^{2,2x}$,
$P^{2y,2}$,
$P^{4,2x}$, and
$P^{2y,4}$ for $x > 1$ and $y > 1$.
We will consider each form separately.

First, we consider quartics of the form $P^{2,2}$ in $\can{a}{b} \cap \BR(i,j)$. We analyse the fragment $A[i \dd n][j \dd j+2^{b})$
and treat it as a metastring by viewing rows as metacharacters.
We observe that each considered quartic defines a prefix that is a square there. Further, all those squares need to be
primitive, as otherwise $P$ could be written as $P=Q^{k,1}$, for some $k>1$, contradicting the primitivity of $P$.
Thus, by \cref{lem:two} we have at most two possible heights for the considered quartics. By a symmetric argument,
we have at most two possible widths, and so at most 4 quartics.

Second, we consider quartics of the form $P^{2,2x}$ in $\can{a}{b} \cap \BR(i,j)$. By the same reasoning as above,
we have at most two possible heights for the considered quartics; let $h$ by one of them.
We analyse the fragment $A[i \dd i+h)[j \dd n]$ and treat it as a metastring by viewing columns as metacharacters.
Each considered quartic with height $h$ corresponds to a prefix
that is a $(2x)$-th power, for some $x>1$. By \cref{lem:three}, all such prefixes are powers
of the same $U$; let $U^{2x}$ be the longest such prefix. Then, for any $x'<x$, the prefix $U^{2x'}$
also occurs at position $(i,j+|U|)$, so the occurrence at position $(i,j)$ cannot be an extreme occurrence. Therefore, for every possible height,
we have at most one quartic, so at most 2 in total.

Third, we consider quartics of the form $P^{4,2x}$ for $x>1$ in $\can{a}{b} \cap \BR(i,j)$. We (again)
analyse the fragment $A[i \dd n][j \dd j+2^{b})$ and treat it as a metastring by viewing its rows as metacharacters.
We observe that each considered quartic defines a prefix that is a primitively rooted fourth power there. 
Thus, by \cref{lem:three} we have at most one possible height, and by the same reasoning as above
at most one quartic.

Symmetric arguments bound the number of quartics of the forms $P^{2y,2}$ and $P^{2y,4}$.
\end{proof}

\begin{lemma}\label{lem:nonthick}
The number of distinct thin quartics in $A$ is $\cO(n^2 \log n )$.
\end{lemma}
\begin{proof}
For each position $(i,j)$ of $A$, $\BR(i,j)$ has a non-empty intersection with at most $\cO(\log n)$ canonical sets due to \cref{lem:aspect_ratio}.
Further, by \cref{lem:thick_period}, there are at most 10 thin quartics in each such intersection.
Since $A$ has $n^2$ positions, the stated bound follows.
\end{proof}
\subsection{Reduction to a Geometric Problem}
We next partition the thick elements of $\BR$ by primitive root.
For each primitive 2D string~$R$, let us denote the thick elements of $\BR$ with primitive root $R$ by $\BR_R$.
Additionally, let us denote by $\OCR(R)$ the set of all positions $(i,j)$ of $A$ where $R^{5,5}$ occurs such that there is an element of $\BR_R$ that fully contains this occurrence of $R^{5,5}$ and has top-left corner equal to $(i-x \cdot \height(R), j-y \cdot \width(R))$ for some non-negative integers $x$ and $y$.

The proof of the following lemma proceeds almost exactly as the proof of Claim 18 in~\cite{CRRWZ}, except that 
we work with  occurrences of $R^{5,5}$ instead of $R^{3,3}$ and do not need the notion of special points.
We provide a detailed description for completeness.

\begin{lemma}[cf.~the proof of~{\cite[Claim 18]{CRRWZ}}]\label{lem:single_area}
For any 2D string $R$, $|\BR_R| \leq |\OCR(R)|$.
\end{lemma}
\begin{proof}
We will map each $Q\in\BR_R$ to an occurrence of $R^{5,5}$ in such a way that two distinct quartic-fragments
$Q,Q'\in\BR_R$ are mapped to distinct occurrences. This will imply that the number of occurrences of $R$ is at least
as large as the number of elements of $\BR_R$. 

For each $x=6,8,\dots$ in this order, we select $Q\in\BR_R$ such that $\height(Q)=x\cdot\height(R)$ and $\width(Q)=y\cdot\width(R)$
is the largest among all $Q'\in\BR_R$ with $\height(Q')=x\cdot\height(R)$. We note that the number of $Q'\in \BR_{R}$ with
$\height(Q')=x\cdot \height(R)$ is at most $y/2-2$, and our goal is to map them to occurrences of $R^{5,5}$ that
have not been used so far. Additionally, we will ensure that those occurrences are all in the same row.
Let $(i,j)$ be the position of an extreme occurrence of $Q$. We observe that $R^{5,5}$ occurs at every position $(i',j')$
with $i'=i+k\cdot \height(R)$ and $j'=j+\ell\cdot \width(R)$, for every $k\in [0\dd x-5]$ and $\ell\in [0\dd y-5]$.
We choose $k\in [0\dd x-5]$ such that none of the occurrences of $R^{5,5}$ at positions $(i+k\cdot \height(R),j+\ell\cdot \width(R))$,
for $\ell=0,1,\ldots,y-3$ have been used so far. This is possible because so far we have used occurrences
of $R^{5,5}$ in only $x/2-3 < x-4$ rows. Then, we map every $Q'\in \BR_{R}$ with $\height(Q')=x\cdot \height(R)$
to an occurrence of $R^{5,5}$ at position $(i+k\cdot \height(R),j+\ell\cdot \width(R))$, for some $\ell\in [0\dd y-5]$,
which is possible due to $y/2-2 \leq y-4$.
\end{proof}

Thus, it remains to upper bound $\sum_R|\OCR(R)|$, i.e., the sum, over all $R$, of the number of occurrences of $R^{5,5}$ which are contained in some element of $\BR_R$.

Consider an occurrence of a 2D string of the form $R^{5,5}$, for a primitive string $R$, at a position $(i,j)$ of $A$.
We call position $(i + 2 \cdot\height(R), j + 2 \cdot\width(R))$ the \emph{anchor} of this occurrence; see Figure~\ref{fig:anchor}.

\begin{figure}[t]
\centerline{\includegraphics[scale=.8]{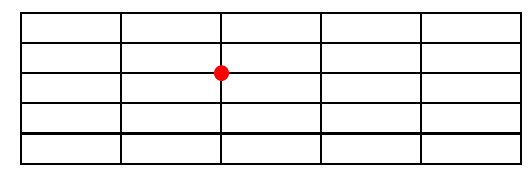}}
\caption{The red point corresponds to the anchor of the shown occurrence of $R^{5,5}$.}
\label{fig:anchor}
\end{figure}

Now, for each primitive string $R$, for each element of $\OCR(R)$, we assign the corresponding occurrence of $R^{5,5}$ to its anchor.
Let $\assign(i,j)$ be the set of primitive 2D strings~$R$ such that occurrences of $R^{5,5}$ have been assigned to position $(i,j)$.
We have
\begin{equation}\label{eq:distr}
\sum_R|\OCR(R)|=\sum_{i=1}^{n}\sum_{j=1}^{n}|\assign(i,j)|.
\end{equation}

It now suffices to show that $\sum_{i=1}^{n}\sum_{j=1}^{n}|\assign(i,j)|=\cO(n^2 \log n)$.
We will in fact show that $|\assign(i,j)|=\cO(\log n)$ for all $i,j$, which straightforwardly yields the desired bound.

Let us fix a position $(i,j)$.
By applying \cref{lem:three} horizontally and vertically one easily obtains the following fact.

\begin{fact}[{\cite[Corollary 13]{CRRWZ}}]\label{fact:threethree}
Let $a$, $b$ be non-negative integers and $W, Z$
be different 2D strings with height in $[2^a \dd 2^{a+1})$ and width in $[2^b\dd 2^{b+1})$.
If $W^{3,3}$ and $Z^{3,3}$ occur at some position of $A$, then at least one of $W$ and $Z$
is not primitive.
\end{fact}

This, together with the fact that, for each $R \in \assign(i,j)$, $R^{3,3}$ occurs at position $(i,j)$, implies the following.

\begin{fact}\label{lem:unique}
For each pair $(a,b) \in [1 \dd \floor{\log n}]^2$, for each position $(i,j)$ of $A$, the set $\assign(i,j)$ contains at most one element with height in $[2^a \dd 2^{a+1})$ and width in $[2^b \dd 2^{b+1})$.
\end{fact}

Let us define a map $\Sigma^{*,*} \rightarrow [1 \dd \floor{\log n}]^2$ as $ g(R) \xmapsto{}{} (\floor{\log (\height(R))}, \floor{\log (\width(R))})$.
Let $f$ be the restriction of $g$ to the domain $\assign(i,j)$.
Due to \cref{lem:unique}, $f$ is an injective function.
We henceforth identify each element $R$ of $\assign(i,j)$ with point~$f(R)$.
We denote the image of $f$ by $\P$.

We say that a point $(a,b) \in \mathbb{Z}^2$ \emph{dominates} a point $(a',b')$ if and only if $a' \leq a$ and $b' \leq b$;
the dominance is said to be \emph{strict} if and only if $a' < a$ and $b' < b$;
the dominance is said to be \emph{strong} if and only if $a' \leq a$ and $b' < b$ or $a' < a$ and $b' \leq b$;
the dominance is said to be \emph{weak} if and only if $a' \leq a$ and $b' \leq b$.
A set of points on which the domination relation forms a total order is called a \emph{chain}.
A set of points such that none dominates another is called an \emph{antichain}.
We are going to use Dilworth's theorem \cite{dilworth}, which states that, in any finite partially ordered set, the size of the largest antichain is equal to the minimum number of chains in which the elements of the set can be decomposed.

For two primitive 2D strings $S$ and $T$,
with $3\cdot \height(S)\leq 2\cdot \height(T)$ and $\width(S)< \width(T)$,
we say that
$S$ \emph{horizontally spans} $T$ when the 2D string $\hspan{S}{T}$, consisting of the $3\cdot \height(S)$ topmost rows of $T^{2,2}$, equals $S^{3,y}$ for some integer $y \geq 2$.
Similarly, when $3\cdot \width(S)\leq 2\cdot \width(T)$ and $\height(S)< \height(T)$, we say that $S$ \emph{vertically spans} $T$ when the 2D string $\vspan{S}{T}$, consisting of the $3\cdot \width(S)$ leftmost columns of $T^{2,2}$, equals $S^{x,3}$ for some integer $x \geq 2$.

\begin{fact}\label{fact:span}
Let $S$ and $T$ be two primitive 2D strings.
If $S$ spans $T$ horizontally, then the horizontal period of $\hspan{S}{T}$ is $\width(S)$.
Symmetrically, if $S$ spans $T$ vertically, then the vertical period of $\vspan{S}{T}$ is $\height(S)$.
\end{fact}
\begin{proof}
We only prove the first statement as the second one follows by symmetry.
Let us view $\hspan{S}{T}$ as a metastring $Z$ by viewing each of its columns as a metacharacter; the horizontal period of $\hspan{S}{T}$ equals $p:=\per(Z)$.
Note that $\width(S)$ is a period of $Z$.
Towards a contradiction, suppose that $p<\width(S)$.
Then, an application of the periodicity lemma to $Z$ implies that $p$ must divide $\width(S)$.
This fact contradicts the primitivity of $S$, as we would have that $S=(S[1 \dd \height(S)][1 \dd p])^k$ for $k=\width(S)/p$.
\end{proof}

When reading the following lemma, one can think of $M$ being in $\assign(i,j)$.
However, the lemma is slightly more general, as needed for the algorithm that is presented in \cref{sec:algo}.

\begin{lemma}\label{lem:span}
Let $R \in \assign(i,j)$ and $M$ be a primitive 2D string such that:
\begin{itemize}
\item $M^{5,5}$ has an occurrence with anchor $(i,j)$;
\item $g(M)$ strictly dominates $g(R)$.
\end{itemize}
Then, $R$ spans $M$ either horizontally or vertically (or both).
\end{lemma}
\begin{proof}
By the definition of $\assign(i,j)$, the occurrence of $R^{5, 5}$ assigned to position $(i,j)$ appears inside an element of $\BR_{R}$, that is, a bottom-right occurrence of a thick quartic with primitive root $R$.
Let us denote this quartic by $Q$.
Observe that the considered occurrence of $Q$ cannot be fully contained inside the occurrence of $M^{4,4}$ at position $(i-2\cdot \height(M),j-2\cdot \width(M))$
as this would contradict the fact that the considered occurrence of $Q$ is bottom-right: there would be another occurrence $\width(M)$ positions to the right.
Therefore, $Q$ must contain at least one of the following four fragments of $A$, depicted in \cref{fig:fourfrag}:
\textcolor{blue}{$A[j \dd j + 3\cdot \height(R))[i \dd i + 2\cdot \width(M))$},
\textcolor{green!70!black}{$A[j \dd j + 3\cdot \height(R))[i - 2 \cdot \width(M) \dd i)$},
\textcolor{red}{$A[j \dd j + 2 \cdot \height(M))[i \dd i + 3\cdot \width(R))$},
\textcolor{black!60!white}{$[j - 2 \cdot \height(M) \dd j)A[i \dd i + 3\cdot \width(R))$}.
\begin{figure}[htpb]
\centering
\includegraphics[width=5cm]{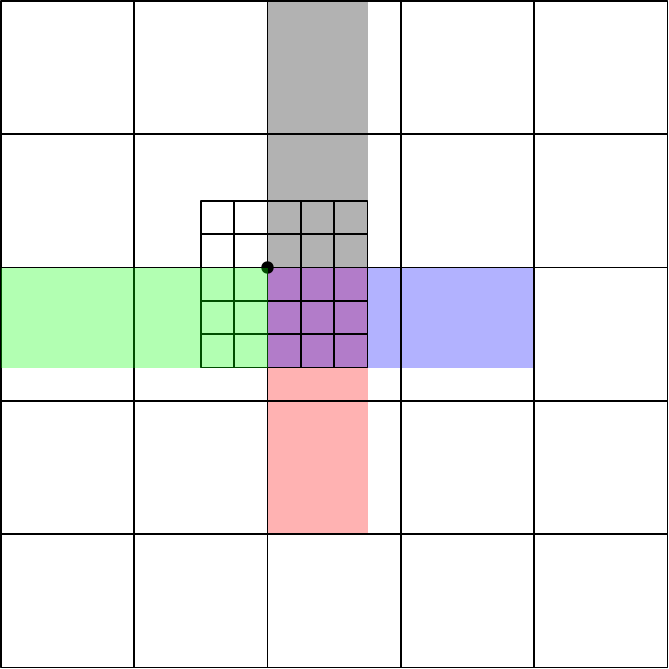}
\caption{The considered occurrences of each of $M^{5,5}$ and $R^{5,5}$ are shown, together with the four specified fragments, at least one of which must be fully contained in $Q$.}\label{fig:fourfrag}
\end{figure}

We next show that in either of the first two cases, $R$ horizontally spans $M$. In the remaining cases, a symmetric argument yields that $R$ vertically spans $M$.
We first argue that, in either of the first two cases, we have $3\cdot \height(R) \leq 2 \cdot \height(M)$.
If this were not the case, $\height(R)$ would be a vertical period of $M^{2,2}$ since $\height(R)+\height(M) \leq 2\height(M)$.
Then, a direct application of Lemma~\ref{lem:FW} to the metastring obtained from $M^{2,2}$ by viewing its rows as metacharacters, would contradict the primitivity of $M$. 
The fact that $\width(R)<\width(M)$ is a direct consequence of $g(M)$ strictly dominating $g(R)$. 
Then, in either of the considered cases, we have that both $\width(M)$ and $\width(R)$ are horizontal periods of $\hspan{R}{M}$.
Hence, by the periodicity lemma applied to the metastring obtained from $\hspan{R}{M}$ by viewing each of its columns as a metacharacter, we have that $\width(R)$ divides $\width(M)$.
This completes the proof.
\end{proof}

\begin{lemma}\label{lem:onechain}
Two primitive 2D strings $R_1$ and $R_2$ such that
$g(R_1)$ and
$g(R_2)$
form an antichain cannot both horizontally span a 2D string $R$.
\end{lemma}
\begin{proof}
Let $g(R_1)=(a_1,b_1)$ and $g(R_2)=(a_2,b_2)$.
Without loss of generality, we can assume that
$a_1 < a_2$ and
$b_1 > b_2$.
Suppose towards a contradiction that both $R_1$ and $R_2$ horizontally span~$R$.
By applying \cref{fact:span}, we obtain that
\begin{itemize}
\item $\width(R_1)$ is the horizontal period of $\hspan{R_1}{R}$ and
\item $\width(R_2)$ is the horizontal period of $\hspan{R_2}{R}$.
\end{itemize}
Note that, for any $k \in [1 \dd \height(R)]$, the horizontal period of the string comprised of the $k$ topmost rows of $R^{1,2}$ equals the least common multiple of the periods of those $k$ rows.
Hence, the period cannot decrease as we increase the number of considered rows.
We thus have $\width(R_2) \leq \width(R_1)$ since our assumption that $a_2>a_1$ implies that $\height(R_2) > \height(R_1)$.
This is a contradiction to our assumption that $b_1>b_2$, which implies that $\width(R_1)>\width(R_2)$.
\end{proof}

The above lemma, \cref{lem:unique}, and Dilworth's theorem together imply the following.

\begin{corollary}\label{cor:twochains}
All primitive 2D strings that span a primitive 2D string $R$ can be decomposed to two sets $H$ and $V$, such that
\begin{itemize}
\item the elements of $H$ span $R$ horizontally;
\item the elements of $V$ span $R$ vertically;
\item the restriction of $g$ to $H \cup V$ is an injective function;
\item each of the sets $g(H)$ and $g(V)$ is a chain.
\end{itemize}
\end{corollary}

Finally, as mentioned in the introduction, we need the following purely geometric lemma that follows from the result of
Marcus and Tardos~\cite{DBLP:journals/jct/MarcusT04} on the number of 1s in an $m\times m$ binary matrix~$M$ that
avoids a fixed permutation $P$ as a submatrix.

\begin{lemma}\label{lem:geom}
Consider a positive integer $m$ and a set $\P \subseteq [1 \dd m]^2$.
If, for each $p \in \P$, the set
of points of $\P$ that are strictly dominated by $p$
can be partitioned into at most two chains,
then $|\P|=\cO(m)$.
\end{lemma}

\begin{proof}
We think of $\P$ as an $m\times m$ matrix $M[1\dd m][1\dd m]$, where $M[a][b]=1$ when $(a,b)\in \P$ and
$M[a,b]=0$ otherwise. Next, we say that $M$ contains a matrix $P$ as a submatrix when $P$ can be obtained
from $M$ by removing rows, removing columns, and changing 1s into 0s. We claim that, by the assumptions
in the lemma, $M$ does not contain the following matrix $P$ as a submatrix:
\begin{center}
\begin{tabular}{ |c|c|c|c| } 
 \hline
 & &  1 & \\ 
 \hline
 & 1 &  & \\ 
 \hline
 1 & &  & \\ 
 \hline
 & &  & 1 \\ 
 \hline
\end{tabular}
\end{center}
To establish this, assume otherwise. Then, there exists $(a,b)\in \P$ and $(a_{1},b_{1}),(a_{2},b_{2}),(a_{3},b_{3})\in \P$
such that $(a,b)$ strictly dominates $(a_{1},b_{1}),(a_{2},b_{2}),(a_{3},b_{3})$ and further
$(a_{1},b_{1}),(a_{2},b_{2}),(a_{3},b_{3})$ create an antichain. By Dilworth's theorem, this implies that
the points in $\P$ dominated by $(a,b)$ cannot be partitioned into two chains, a contradiction.
Thus, $M$ indeed does not contain $P$ as a submatrix. Because $P$ is a permutation matrix,
this implies $|\P|=\cO(m)$.
\end{proof}

We now complete the proof of our main result with the aid of \cref{lem:geom}.

\begin{theorem}\label{thm:main}
An $n \times n$ 2D string has $\cO(n^2 \log n)$ distinct quartics.
\end{theorem}
\begin{proof}
The number of distinct quartics is at most $|\BR|$.
We then have
\begin{align*}
|\BR| & =\cO(n^2 \log n) + \sum_R |\BR_R| & \text{(\cref{lem:thick_period,lem:nonthick})} \\
    & = \cO(n^2 \log n) + \sum_R |\OCR(R)| & \text{(\cref{lem:single_area})} \\
    & = \cO(n^2 \log n) + \sum_{i=1}^{n}\sum_{j=1}^{n}|\assign(i,j)|. & \text{(\ref{eq:distr})}
\end{align*}
To conclude the proof, it remains to show that $|\assign(i,j)|=\cO(\log n)$ for all $(i,j)\in [1\dd n]^2$.
Let $m=\lfloor\log n\rfloor$, and
recall that $\P \subseteq [1 \dd m]^2$ was defined as the image of $f$, which in turn was the restriction of $g$ to
the domain $\assign(i,j)$.
By \cref{lem:unique}, we only need to show that $|\P|=\cO(m)$.
By \cref{lem:span} and \cref{cor:twochains}, for each $p\in \P$, the set of all points of $\P$
that are strictly dominated by $p$ can be partitioned into at most two chains.
Thus, by \cref{lem:geom} we conclude that indeed $|\P|=\cO(m)$, concluding the proof.
\end{proof}

\section{The Optimal Algorithm}
\label{sec:algo}

In this section, we consider an $n\times n$ 2D string $A$ over an ordered alphabet $\Sigma$ and describe our $\cO(n^2 \log n)$-time algorithm for computing distinct quartics in $A$.
After an $\cO(n^2 \log n)$-time preprocessing, that consists of sorting the characters and renaming them, we can assume that the entries of $A$ are in $[1 \dd n^{2}]$.

\subsection{Technical Overview}
\label{sec:over}
Our algorithm computes thin and thick quartics separately.
Here, we provide an overview of these computations.

\subsubsection*{Computation of Thin Quartics}

For thin quartics, our algorithm is quite similar to the combinatorial analysis.
For each position $(i,j)$, we compute an $\cO(\log n)$-size superset $\mathcal{C}$ of the canonical sets that have a non-empty intersection with $\BR(i,j)$.
We do this by relating extreme occurrences of quartics with occurrences of squares in metastrings obtained by viewing the columns of $A[i \dd i+2^a)[1 \dd n]$, where $a \in [1 \dd \floor{\log n}]$, as metacharacters.
These squares can be efficiently computed and give us a handle on the sought thin quartics.
Then, we compute the intersection of each canonical set in $\mathcal{C}$ with $\BR(i,j)$ in constant time
using known tools that allow us to efficiently operate on the metastrings. We do this by fixing $(a,b) \in [1 \dd \floor{\log n}]^2$ and computing all quartics with height in $[2^a \dd 2^{a+1})$ and width $[2^b \dd 2^{b+1})$ that occur at position $(i,j)$ of $A$, of the form $P^{2y,2}$, $P^{2y,4}$, $P^{2,x}$, and $P^{4,4x}$, for a primitive 2D string $P$ and $x,y\geq 1$.
A detailed description can be found in~\cref{ComputeThin}.

\subsubsection*{Computation of Thick Quartics}

For thick quartics, our algorithmic approach follows our combinatorial approach in a more relaxed sense.
The main technical challenge is to compute, for each position $(i,j)$, an $\cO(\log n)$-size set $\mathcal{R}$ of primitive 2D strings $R$, such that
$\assign(i,j) \subseteq \mathcal{R}$.
Then, those supersets can be postprocessed as in \cite{CRRWZ} in time linear in their total size to yield the sought distinct thick quartics.
The computation of $\mathcal{R}$ is split into two major steps which we outline next. A detailed description can be found in~\cref{ComputeThick}.

\paragraph*{Skyline Computation.}
First, we compute a set $\S$ of \emph{skyline primitive 2D strings} such that $S\in \mathcal{S}$ when
(a) $S^{5,5}$ has an occurrence anchored at position $(i,j)$, and (b) there is no other primitive 2D string $T$ with $g(T)\geq g(S)$ such that $T^{5,5}$ has an occurrence anchored at position $(i,j)$.
This part of the proof is quite technical:
it heavily relies on the analysis of periodicity for 1D (meta)strings and, roughly speaking, on the analysis of the evolution of the horizontal periodic structure of a 2D string as rows are appended to it.
We show that, for each $a \in [1 \dd \floor{\log n}]$, there is a single candidate $h \in [2^a \dd 2^{a+1})$ to be considered as the height of an element of $\mathcal{S}$.
Then, using runs in 1D metastrings whose origins in $A$ have sufficient overlap and bit-tricks, we can compute the widest 2D string $S$ with height $h$ such that an occurrence of $S^{5,5}$ is anchored at position $(i,j)$ in constant time (using batched computations), if one exists.

\paragraph*{Computation of Dominated 2D strings.}
This turned out to be the most challenging part of our approach.
For this exposition, let us treat $\Sigma^{*,*}$ as a partially ordered set, in the order of decreasing widths.
Let $\S=\{S_1, \ldots , S_\ell\}$ in accordance with this order.
To obtain $\mathcal{R}$ from $\S$, we need to add to it the 2D strings $R \in \assign(i,j) \setminus \S$. By the construction of $\S$ we know
that there exists $S\in \S$ such that $g(R) \leq g(S)$.

Our combinatorial analysis implies that each $R \in \assign(i,j)$ spans each element $S \in \mathcal{S}$ for which $g(S)$ stricly dominates $g(R)$ either vertically or horizontally.
It turns out that if $R$ spans all of these elements of $\S$ either vertically or horizontally, it is easy to compute it efficiently.
This is, unfortunately, not the case in general.
However, we observe that $R$ spans vertically (resp.~horizontally) a contiguous subset of $\S$.
We show that the problem boils down to computing the union, over all~$k$, of sets~$I_k$, where $I_k$ contains exactly those primitive 2D strings that span $S_{k-1}$ vertically and span~$S_k$ horizontally.
Crucially, we observe that due to a strong form of transitivity of the spanning property, the intersection of any two such $I_k$ and $I_{k'}$ consists of a number of the smallest elements of both (i.e., their longest common prefix if viewed as strings).
Hence, by computing the elements of $I_k$ from the largest to the smallest, we can stop whenever we encounter an element that has already been reported by this procedure.
This allows us to reduce the computation of $\mathcal{R}$ to the problem of efficiently computing sets $I_k$.
To this end, we prove that an $\cO(\log n)$-bits representation of the evolution of the periodic structure of certain fragments of $A$ as rows and columns are appended to them suffices for inferring $I_k$; we then use tabulation to infer it efficiently.

\subsection{Preprocessing}

For each $a \in [1 \dd \floor{\log n}]$, for each $i \in [1 \dd n-2^a]$, let us denote by $\hmeta{i}{a}$ the metastring obtained from $A[i \dd i+2^a)[1 \dd n]$ by viewing each column as a metacharacter; let us denote the collection of those strings by $\mathcal{H}$. 
Similarly, for each $b \in [1 \dd \floor{\log n}]$, for each $j \in [1 \dd n-2^b]$, let us denote by $\vmeta{j}{b}$ the metastring obtained from $A[1 \dd n][j \dd j+2^b)$ by viewing each row as a metacharacter; let us denote the collection of those strings by $\mathcal{V}$. See Figure~\ref{fig:H-V}.
Overall, we have $\cO(n\log n)$ metastrings, each of length $n$.
The metastrings in $\mathcal{H}$ and $\mathcal{V}$ can be computed in $\cO(n^2 \log n)$ time in total as follows.
It suffices to explain how to compute all strings of the form $\hmeta{i}{a}$ as the computation of the strings of the form $\vmeta{j}{b}$ is symmetric.
For any $i$, $\hmeta{i}{1}$ can be computed by radix-sorting two-tuples of characters and renaming them, while $\hmeta{i}{a}$ for $a>1$ can be similarly computed using two-tuples consisting of characters from $\hmeta{i}{a-1}$ and $\hmeta{i+2^{a-1}}{a-1}$.

\begin{figure}[htpb]
\centering
\includegraphics[scale=0.4]{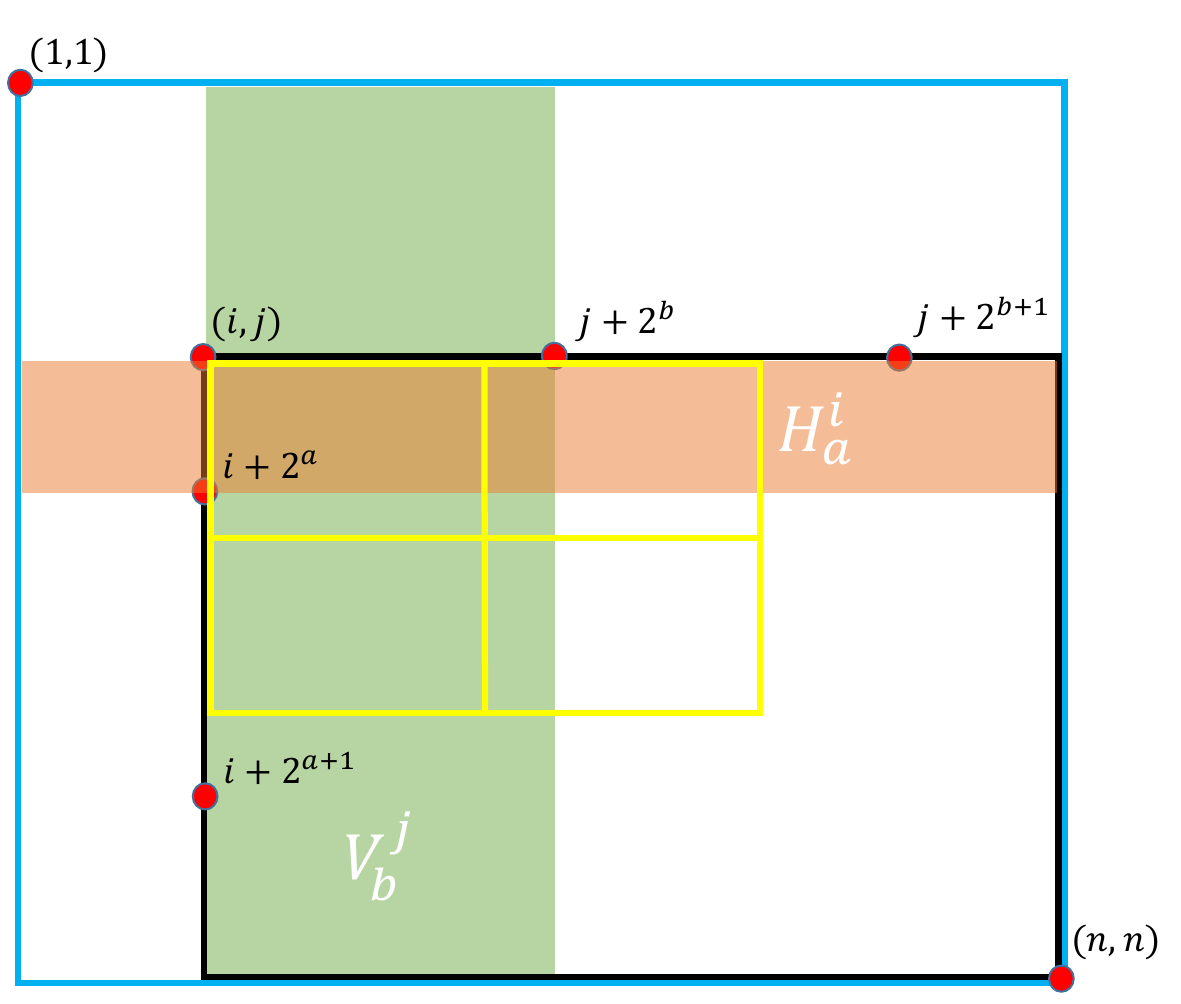}
\caption{$A$ is presented in blue. The black rectangle corresponds to $A[i\dd n][j\dd n]$. The orange rectangle corresponds to $\hmeta{i}{a}$. The green rectangle corresponds to $\vmeta{j}{b}$. A quartic in $\BR(i,j)\cap\can{a}{b}$ is shown in yellow.}\label{fig:H-V}
\end{figure}

All runs in a string of length $n$ over an alphabet of integers from $[n^{c}]$, for some constant $c$ (for short: polynomial alphabet),
can be computed in $\cO(n)$ time~\cite{KK:99,ellert_et_al:LIPIcs.ICALP.2021.63}.
We compute the runs in each of the $\cO(n \log n)$ length-$n$ strings in $\mathcal{H}$ and $\mathcal{V}$ in $\cO(n^2 \log n)$ time in total.

Let us now define the following primitive operations for fragments of a text $T$.
\begin{itemize}
\item An \emph{internal pattern matching} query takes as input two fragments $U$ and $V$ of a text, and returns all occurrences of $U$ in $V$ in the form of $|V|/|U|$ arithmetic progressions with difference $\per(U)$; such a query is denoted by $\IPM(U,V)$.
\item A \emph{2-period} query takes as input a substring $U$ of a text, checks if $U$ is periodic and, if so, it also returns $U$'s period.
\item A \emph{periodic extension query} takes as input a periodic fragment of a text and returns the (unique) run that has the same period and contains it~\cite{tomeksthesis}.
\end{itemize}

We preprocess each of the strings in $\mathcal{H}$ and $\mathcal{V}$ according to the following theorem, in $\cO(n^2 \log n)$ time in total.

\begin{theorem}[\cite{DBLP:conf/soda/KociumakaRRW15,tomeksthesis}]\label{thm:IPM}
A string $T$ of length $n$ over a polynomial alphabet can be preprocessed in $\cO(n)$ time so that queries of the form $\IPM(U,V)$, for any two fragments $U$ and $V$ of $T$, can be answered in $\cO(|V|/|U|)$ time, and 2-period queries and period extension queries can be answered in $\cO(1)$ time.
\end{theorem}

We also apply the following lemma to each of the strings in $\mathcal{H}$ and~$\mathcal{V}$ in $\cO(n^2 \log n)$ time in total.

\begin{lemma}\label{lem:long_sq}
Given a string $T$ of length $n$ over a polynomial alphabet, we can compute the length of the longest square that occurs at each position of $T$ in $\cO(n)$ time.
In addition, a query of the form $(i,x)$, asking for the primitive squares that occur at some position $i$ and have length in $[2^x \dd 2^{x+1})$ can be answered in $\cO(1)$ time after an $\cO(n)$-time preprocessing.
\end{lemma}
\begin{proof}
We first compute all distinct squares in $T$ in $\cO(n)$ time.~\cite{BannaiIK17,CrochemoreIKRRW14}.
Then, we build the suffix tree of $T$ in $\cO(n)$ time~\cite{DBLP:conf/focs/Farach97}, with nodes weighted by string-depth.
We say that a node $v$ is the \emph{weighted ancestor} of a node $u$ at depth $\ell$ if $v$ is the highest ancestor of $u$ with weight of at least $\ell$.
Belazzougui et al.~\cite{DBLP:conf/cpm/BelazzouguiKPR21} have shown that the suffix tree of a length-$n$ string can be preprocessed in $\cO(n)$ time so that weighted ancestor queries (with weights equal to string-depths) can be answered in $\cO(1)$ time.
After this preprocessing, we compute the locus of each square (in decreasing order with respect to length) in the suffix tree of $T$ in $\cO(1)$ time in total and make the corresponding node explicit. This takes $\cO(n)$ time in total.
Additionally, we mark the nodes corresponding to primitive squares with blue and the nodes corresponding to non-primitive squares with red---we can distinguish between the two cases in $\cO(1)$ time using a 2-period query (see \cref{thm:IPM}).

Then, the length of the longest square that occurs at a position $i$ is equal to the string-depth of the nearest ancestor of the node with path-label $T[i \dd n]$ that is marked (either red or blue).
Similarly, the lengths of the primitive squares that occur at some position $i$ and have length in $[2^x \dd 2^{x+1})$ can be computed in constant time given the string-depths of the at most two nearest ancestors of the node $v$ with path-label $T[i \dd i+2^{x+1}-1)$ that are marked with blue; node $v$ can be computed in $\cO(1)$ time with a weighted ancestor query.
The nearest blue or red marked node for each node in a tree can be precomputed in $\cO(n)$ time. This concludes the proof.
\end{proof}

\subsection{Computation of Thin Quartics}
\label{ComputeThin}
\subsubsection{Computation of $\cO(\log n)$ Canonical Sets}

The following observation and fact are counterparts of \cref{lem:aspect_ratio} that are easier to exploit computationally.

\begin{observation}\label{obs:sq}
An occurrence of a quartic $Q$ with height in $[2^a \dd 2^{a+1})$, where $a \in \mathbb{Z}_{+}$,
at a position $(i,j)$ of $A$ implies an occurrence of a square of length $\width(Q)$ at position $j$ in string $\hmeta{i}{a}$.
\end{observation}

\begin{fact}\label{fct:sq}
An occurrence of a square with length in $[2^b \dd 2^{b+1})$, where $b \in \mathbb{Z}_{+}$, at some position~$j$ in string $\hmeta{i}{a}$, where $a, i \in \mathbb{Z}_{+}$, implies that no quartic with height in $[2 \dd 2^{a-1})$ and width in $[2 \dd 2^{b-1})$ can have an extreme occurrence at position $(i,j)$ of $A$.
\end{fact}
\begin{proof}
Let us denote the specified square by $UU$.
The claim follows by the observation that any quartic with height in $[2 \dd 2^{a-1})$ and width in $[2 \dd 2^{b-1})$ that occurs at position $(i,j)$ of $A$
must also have an occurrence at position $(i,j+|U|)$.
\end{proof}

We next combine \cref{obs:sq,fct:sq} to efficiently compute a small superset of the canonical sets that have a non-empty intersection with $\BR(i,j)$ for a fixed position $(i,j)$.

\begin{lemma}
\label{lem:canonicalsets}
After an $\cO(n^2 \log n)$-time preprocessing of $A$, we can compute, for any position $(i,j)$ of $A$, an $\cO(\log n)$-size superset of the canonical sets that have a non-empty intersection with $\BR(i,j)$ in $\cO(\log n)$ time.
\end{lemma}
\begin{proof}
Let us fix a position $(i,j)$.
For each $a \in [1 \dd \floor{\log n}]$, let $\ell_a$ be the length of the longest square at position $j$ in $\hmeta{i}{a}$; all of these lengths can be retrieved in $\cO(\log n)$ time after an $\cO(n^2 \log n)$-time preprocessing due to \cref{lem:long_sq}.
The sequence $\ell_1, \ldots, \ell_{\floor{\log n}}$ is non-increasing since, for each square in $\hmeta{i}{x}$, where $x \in (1 \dd \floor{\log n}]$, there is a square of the same length at the same position in $\hmeta{i}{y}$ for each $y \in [1 \dd x)$.
Now, let $\alpha$ and $\beta$ be integers such that the canonical set $\can{\alpha}{\beta}$ has a non-empty intersection with $\BR(i,j)$.
By \cref{obs:sq}, $(\alpha,\beta)$ must by dominated by some point in $Z = \{(a, \floor{\log \ell_a}): a \in [1 \dd \floor{\log n}]\}$.
Let $Z'$ be the union of $Z$ with the set of all points that are weakly dominated by some point in $Z$, but are not strongly dominated by any point in $Z$; the elements of $Z'$ form a staircase of size $\cO(\log n)$.
Then, by \cref{fct:sq}, $(\alpha,\beta)$ must be in the set $\{(a-x,b-y) : (a,b) \in Z' \text{ and } x,y \in \{0,1,2\}\}$, which is of size $\cO(\log n)$ and can be naively computed in $\cO(\log n)$ time given $\ell_1, \ldots, \ell_{\floor{\log n}}$. See Figure~\ref{fig:lemma25}.
\end{proof}

\begin{figure}[htpb]
\centering
\includegraphics[scale=0.3]{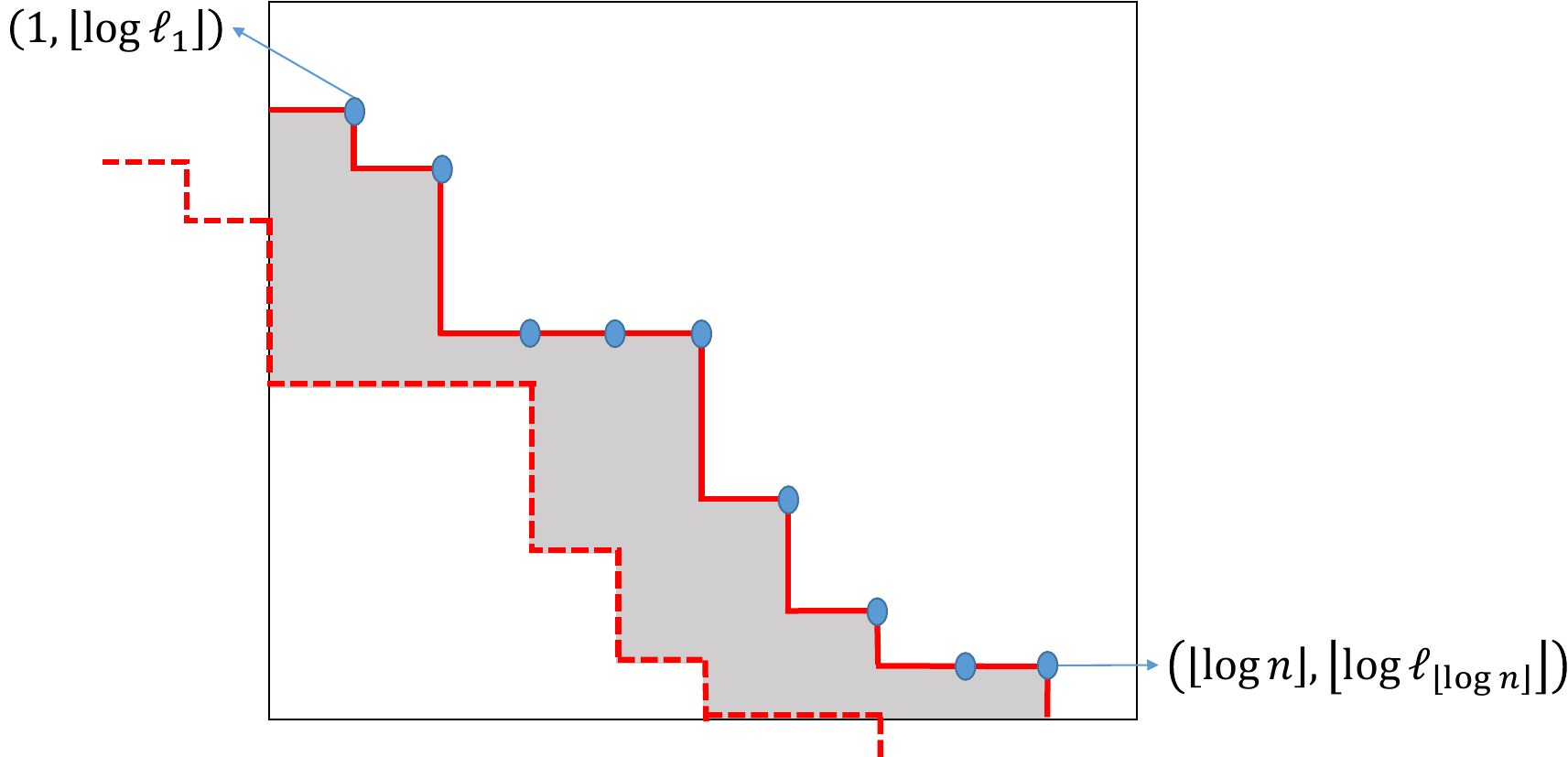}
\caption{The elements of $Z$ are the blue points. The elements of $Z'$ are the points on the solid red line. The dashed red line corresponds to the set of points $\{(a-2,b-2) : (a,b) \in Z'\}$. The set $\{(a-x,b-y) : (a,b) \in Z' \text{ and } x,y \in \{0,1,2\}\}$ is the grey area.}\label{fig:lemma25}
\end{figure}
\subsubsection{Thin Quartics from a Canonical Set}
Fix $(a,b) \in [1 \dd \floor{\log n}]^2$ such that $(a,b)$ corresponds to a canonical set that was computed earlier by~\cref{lem:canonicalsets}.
We wish to find all thin quartics with height in $[2^a \dd 2^{a+1})$ and width $[2^b \dd 2^{b+1})$ that occur at position $(i,j)$ of $A$, recall that we have at most 10 such quartics by Lemma~\ref{lem:thick_period}.
That is, those quartics $Q$, such that $g(Q)=(a,b)$ and $Q$ equals one of $P^{2y,2}$, $P^{2y,4}$, $P^{2,x}$, and $P^{4,4x}$, for a primitive 2D string $P$ and $x,y\geq 1$ (so $P^{2,2}$ will be counted twice).
Below, we show how to compute all quartics that fall in the first two cases, i.e., they are of the form $P^{2y,2}$, $P^{2y,4}$; the remaining ones can be computed symmetrically.

\paragraph{Computation of Three Candidate Widths.}

Let us set $S:=\hmeta{i}{a}$.
\begin{itemize}
\item Similarly to \cref{obs:sq}, an occurrence of a quartic $Q\in \can{a}{b}$ of the form $P^{2y,2}$ at position $(i,j)$, where $P$ is a primitive 2D string, implies an occurrence of a primitive square of length $\width(Q)$ at position $j$ of $S$. At most two primitive squares with length in $[2^b \dd 2^{b+1})$ start at any position $j$ of $S$, cf.~\cref{lem:two}.
\item Analogously, an occurrence of a quartic $Q\in \can{a}{b}$ of the form $P^{2y,4}$ at position $(i,j)$, where $P$ is a primitive 2D string, implies an occurrence of a string $U^4$, where $U$ is a primitive string, at position $j$ of $S$.
There is at most one primitive string $U$ such that $U^4$ occurs at position $j$ of~$S$ and $4\cdot |U| \in [2^b \dd 2^{b+1})$, due to the periodicity lemma (\cref{lem:FW}), since $|U|<2^{b-1}$ is a period of $S[j \dd j+2^b)$.
\end{itemize}

The above analysis implies that there are three candidate widths for quartics in $\BR(i,j) \cap \can{a}{b}$. We next show how to compute them efficiently.
The primitive squares that occur at position $j$ of $S$ and have length in $[2^b \dd 2^{b+1})$ are computed in $\cO(1)$ time due to  \cref{lem:long_sq}.
Additionally, we compute a primitive string~$U$ such that $U^4$ occurs at position $j$ of $S$ and $4\cdot |U| \in [2^b \dd 2^{b+1})$, if one exists, as follows.
We first perform, in $\cO(1)$ time, a 2-period query for $S[j \dd j+2^b)$. If this fragment is not periodic there is no such $U$. Otherwise, we have a candidate length for $U$. We check if $U^4$ occurs at position~$j$ of $S$ by performing a periodic extension query for $S[j \dd j+2^b)$.
Thus, in $\cO(n^2 \log n)$ time in total, we obtain three candidate widths for each of the $\cO(n^2 \log n)$ considered intersections of some $\BR(i,j)$ and some canonical set $\can{a}{b}$.

\paragraph{Processing each Candidate Width.}
We next show how to process one of the three candidate widths, say $w$, in time $\cO(1)$.
Observe that $G:=A[i \dd i + 2^{a-1})[j \dd j + 2^{b})$ has to be contained in the top half of the sought quartic(s).
We start by computing $G$'s occurrences in $A[i \dd i + 2^{a+1})[j \dd j + 2^{b})$.
This can be done in $\cO(1)$ time by employing an IPM query for fragments of $\vmeta{j}{b}$.
It would then be enough to consider, for each occurrence $A[i+h/2 \dd i + h/2 + 2^{a-1})[j \dd j + 2^{b})$ of $G$ the integer $h$ as a candidate height for the sought quartic(s).
Given two integers $w$ and $h$, we can check whether there is a quartic with height $h$ and width $w$ at position $(i,j)$ in $\cO(1)$ time by checking whether the following equalities hold:
    \begin{itemize}
        \item $A[i \dd i + h/2)[j \dd j + 2^{b})=A[i+h/2 \dd i + h)[j \dd j + 2^{b})$;
        \item $A[i \dd i + h/2)[j+w-2^b \dd j + w)=A[i+h/2 \dd i + h)[j+w-2^b \dd j + w)$;
        \item $A[i \dd i + 2^a)[j \dd j + w/2)=A[i \dd i + 2^a)[j+w/2 \dd j + w)$;
        \item $A[i+h-2^a \dd i + h)[j \dd j + w/2)=A[i+h-2^a \dd i + h)[j+w/2 \dd j + w)$.
    \end{itemize}
The latter can be done in $\cO(1)$ time using $4$ IPM queries, one for each of the fragments: $\vmeta{j}{b}$, $\vmeta{j+w-2^b}{b}$, $\hmeta{i}{a}$, and $\hmeta{i+h-2^a}{a}$. Note that, if the first two equalities hold then the two left quarters of the sought quartic are equal and the two right quarters of the sought quartic are equal, see Figure~\ref{fig:check_quartic1}. Moreover, if the last two equalities hold then the two upper quarters of the sought quartic are equal and the two lower quarters of the sought quartic are equal, see Figure~\ref{fig:check_quartic2}.

\begin{figure}[t!]
\begin{center}
\begin{subfigure}[b]{0.99\textwidth}
\centering
\includegraphics[height=6cm]{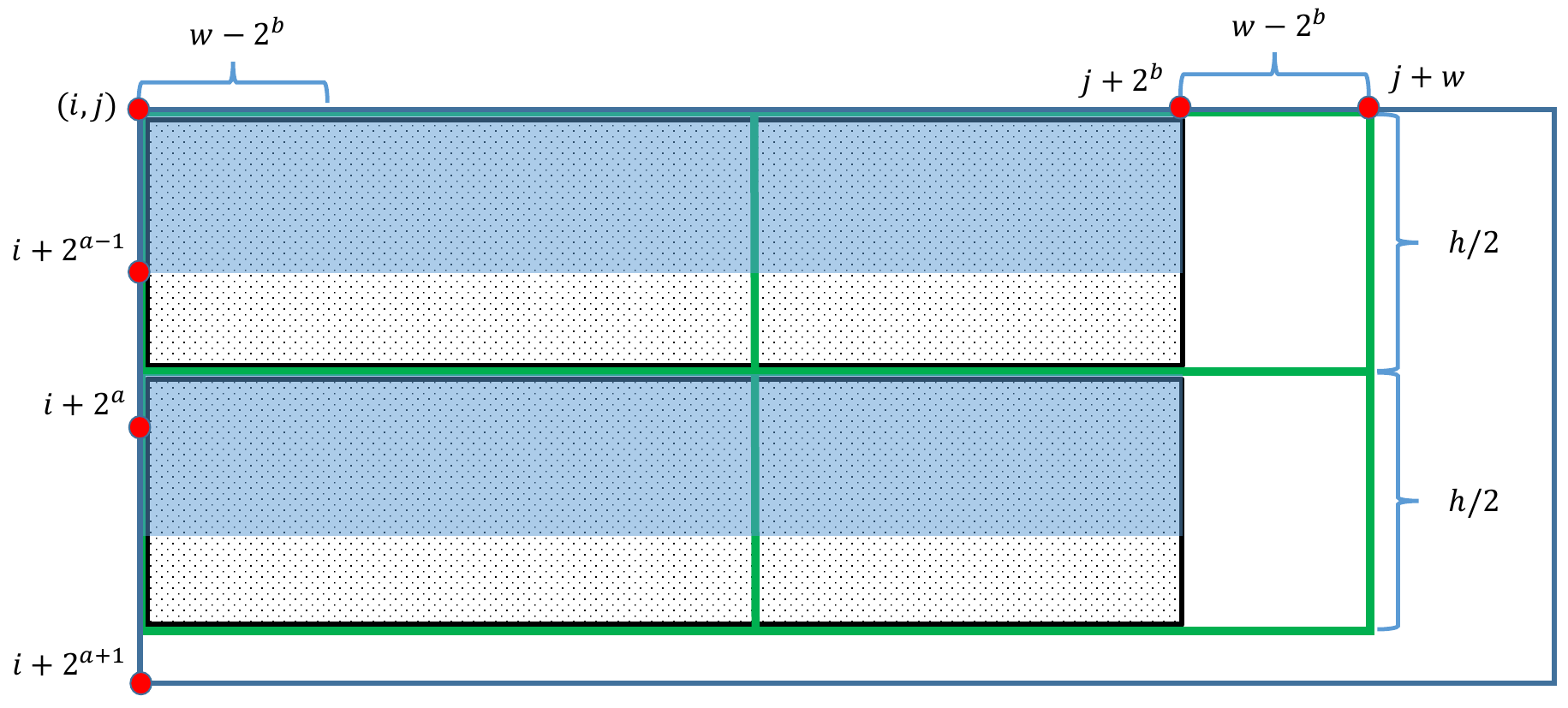}
\caption{The black textured rectangles correspond to $A[i \dd i + h/2)[j \dd j + 2^{b})$ and $A[i+h/2 \dd i + h)[j \dd j + 2^{b})$.}
\end{subfigure}
  \begin{subfigure}[b]{0.99\textwidth}
  \centering
\includegraphics[height=6cm]{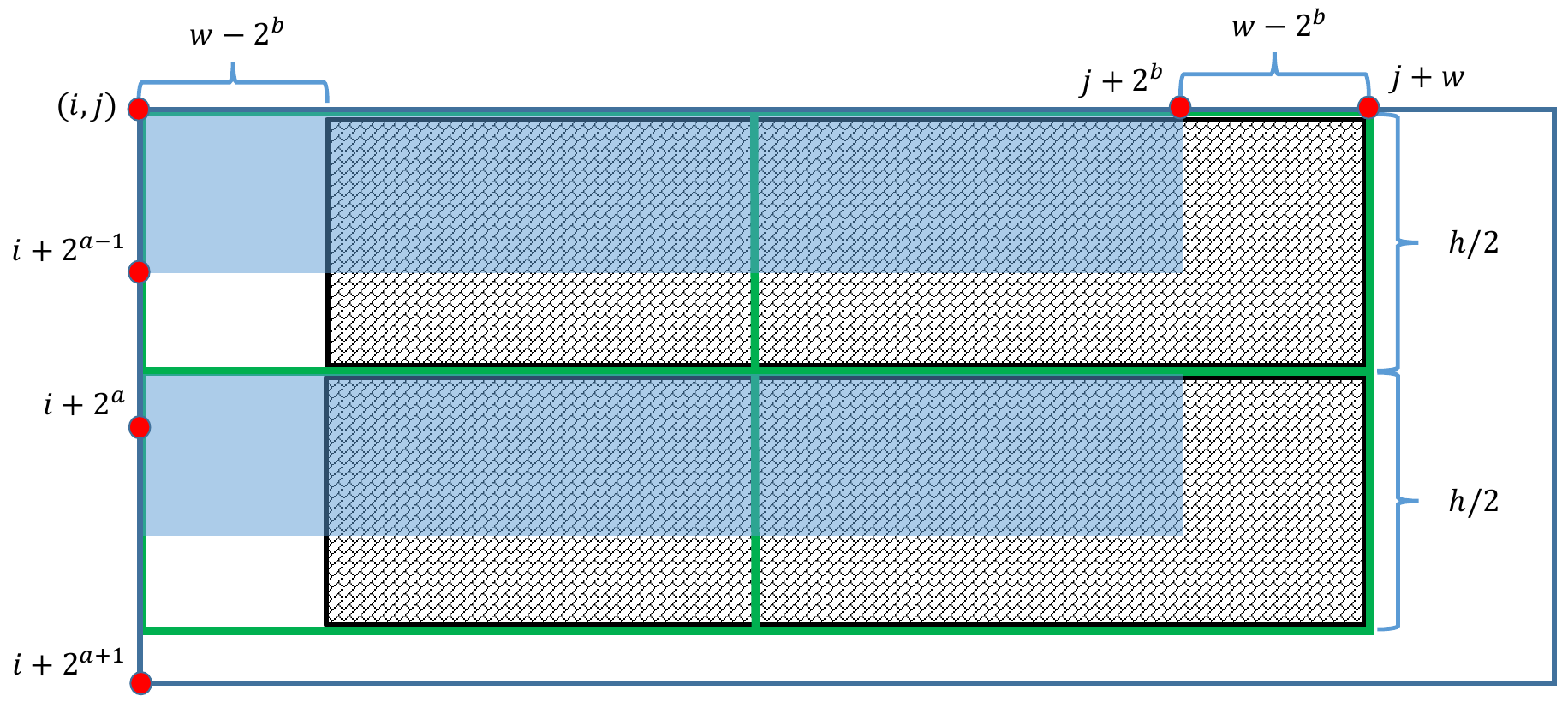}
\caption{The black textured rectangles correspond to $A[i \dd i + h/2)[j+w-2^b \dd j + w)$ and $A[i+h/2 \dd i + h)[j+w-2^b \dd j + w)$.}
\end{subfigure}
\end{center}
\caption{In both (a) and (b) $G:=A[i \dd i + 2^{a-1})[j \dd j + 2^{b})$ is the top rectangle filled with blue that equals to $A[i+h/2 \dd i + h/2+2^{a-1})[j \dd j + 2^{b})$, i.e., the bottom rectangle filled with blue. The sought quartic of height $h$ and width $w$ is shown in green.}
\label{fig:check_quartic1}
\end{figure}
\begin{figure}[t!]
\begin{center}
\begin{subfigure}[b]{0.99\textwidth}
\centering
\includegraphics[height=6cm]{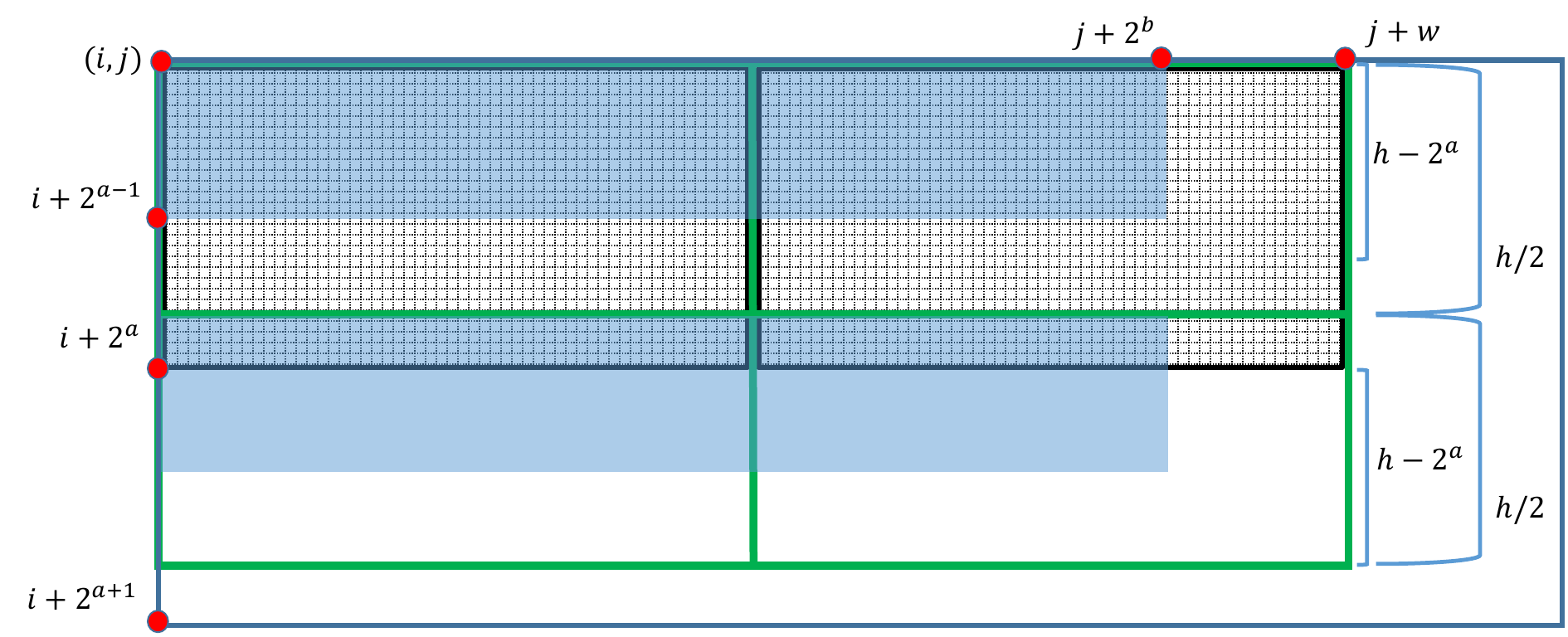}
\caption{The black textured rectangles correspond to $A[i \dd i + 2^a)[j \dd j + w/2)$ and $A[i \dd i + 2^a)[j+w/2 \dd j + w)$.}
\end{subfigure}
  \begin{subfigure}[b]{0.99\textwidth}
  \centering
\includegraphics[height=6cm]{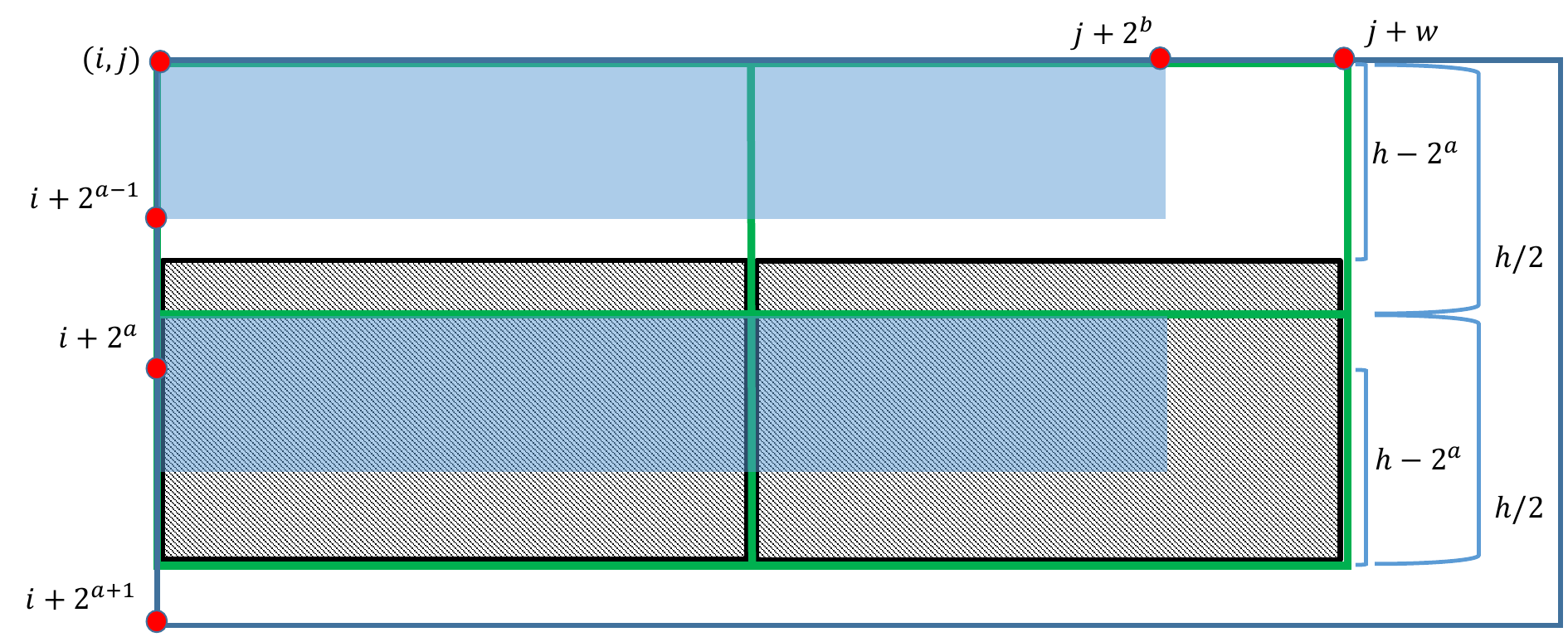}
\caption{The black textured rectangles correspond to $A[i+h-2^a \dd i + h)[j \dd j + w/2)$ and $A[i+h-2^a \dd i + h)[j+w/2 \dd j + w)$.}
\end{subfigure}
\end{center}
\caption{In both (a) and (b) $G:=A[i \dd i + 2^{a-1})[j \dd j + 2^{b})$ is the top rectangle filled with blue that equals to $A[i+h/2 \dd i + h/2+2^{a-1})[j \dd j + 2^{b})$, i.e., the bottome rectangle filled with blue. The sought quartic of height $h$ and width $w$ is shown in green.}
\label{fig:check_quartic2}
\end{figure}

However, if $G$ has a small vertical period, it may have more than a constant number of occurrences in $A[i \dd i + 2^{a+1})[j \dd j + 2^{b})$.
One can still process all of them in constant time, as we describe in detail below.
In what follows, we do not explicitly mention that before reporting any quartic obtained from either of the cases, we first verify that it is in $\can{a}{b}$, i.e., its height is in $[2^a\dd2^{a+1})$ and its width is in $[2^b\dd2^{b+1})$.
\begin{enumerate}
\item If the IPM query returns at most 10 occurrences, we verify the candidate height $h$ in $\cO(1)$ time as explained above.
\item Suppose that the IPM query returns more than 10 occurrences. These occurrences are represented as a constant number of arithmetic progressions such that the difference equals the vertical period of $G$, which we denote by $p$.
Note that, if we consider $G$ as a metastring by viewing its rows as metacharacters then each arithmetic progression corresponds to a run with period $p$ of this metastring.

\begin{enumerate}
\item Let us first focus on the case where $p$ is not a vertical period of the sought quartic(s).
We obtain $\cO(1)$ candidate heights from the returned arithmetic progressions as follows, distinguishing between two cases:
\begin{enumerate}
\item If the top-left quarter of the sought quartic(s) has vertical period $p$, we argue next that it suffices to consider the differences between $i$ and the first element of each other (maximal) arithmetic progression with difference $p$, denoted by $h/2$, i.e., the height of the quarter(s), and as a result, we consider $h$ as a candidate height of the quartic(s). See Figure~\ref{fig:height-ai}.
Suppose that there is such a quartic $A[i \dd i+h)[j \dd j + w)$.
Observe that the occurrences of $G$ at positions $(i,j)$ and $(i+h/2,j)$ must belong to different arithmetic progressions, as otherwise, the whole sought quartic(s) would have vertical period $p$.
Then, the runs corresponding to these arithmetic progressions must overlap. However, two runs with period $p$ cannot overlap by $p$ or more positions and it hence suffices to consider only the first element of the arithmetic progression that corresponds to the rightmost of the two considered runs.
\item In the complementary case it suffices to consider as candidate heights the differences between the last element of the arithmetic progression
that contains position $i$ in~$\vmeta{j}{b}$ and the last element of every other returned arithmetic progression.
This is sufficient because by the assumption on $p$ not being a vertical period of the top-left quarter
the last element of the former arithmetic progression corresponds to an occurrence of $G$ that is fully within the top-left quarter.

\end{enumerate}
In the end, we verify each candidate height as in the case where the IPM query returns at most 10 occurrences.
\item Finally, we need to consider the case where the sought quartic(s) has/have a vertical period equal to $p$.
We already know the candidate horizontal period, which is either $w/2$ or $w/4$, depending on whether the origin of the candidate width $w$ is a primitive square or a string $U^4$ for a primitive $U$.
This means that we already know the primitive 2D string $P$ for which we look for quartic(s) of the form $P^{2y,x}$, where $x \in\{2,4\}$.
Next, we employ periodic extension queries for $\vmeta{j}{b}$ and $\vmeta{j+w-2^b}{b}$ to compute the maximum~$t$ such that $A[i \dd t)[j \dd j + w)$ has a vertical period equal to $p$.
Let $y_{\max}:=\floor{t/2p}$.
Then, the quartics of the form $P^{2y,x}$ that occur at position $(i,j)$ are those for $y\in [1 \dd y_{\max}]$.
Observe that only the quartic of maximal height among these quartics may have an extreme occurrence at position $(i,j)$, as all the other ones have another occurrence at position $(i+p,j)$. See Figure~\ref{fig:thin-per}.
\end{enumerate}
\end{enumerate}

\begin{figure}[htpb]
\centering
\includegraphics[scale=0.4]{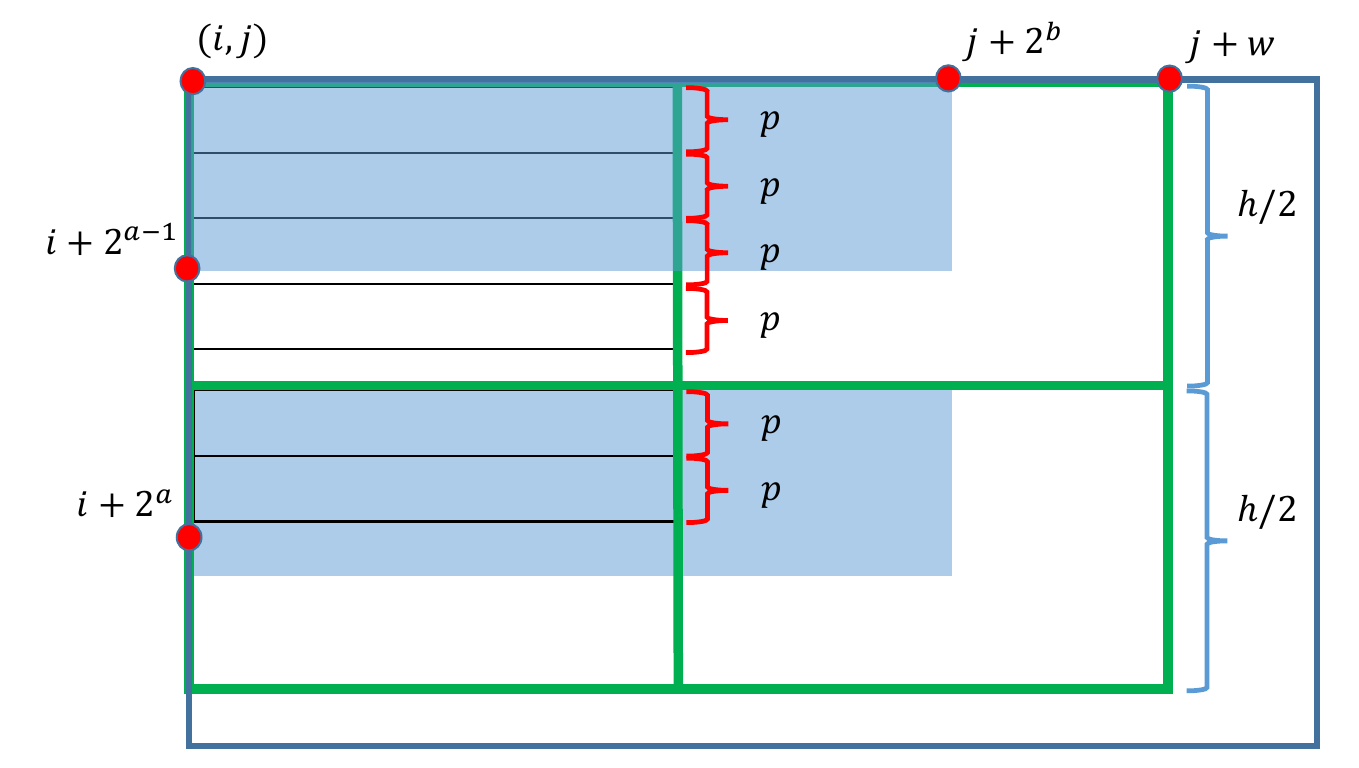}
\caption{$G:=A[i \dd i + 2^{a-1})[j \dd j + 2^{b})$ is the top rectangle filled with blue and a candidate quartic $A[i \dd i+h)[j \dd j + w)$ is shown in green, where the top-left quarter has vertical period $p$. The height $h/2$ corresponds to the difference between $i$ and some occurrence of $G$.}\label{fig:height-ai}
\end{figure}

\begin{figure}[htpb]
\centering
\includegraphics[scale=0.3]{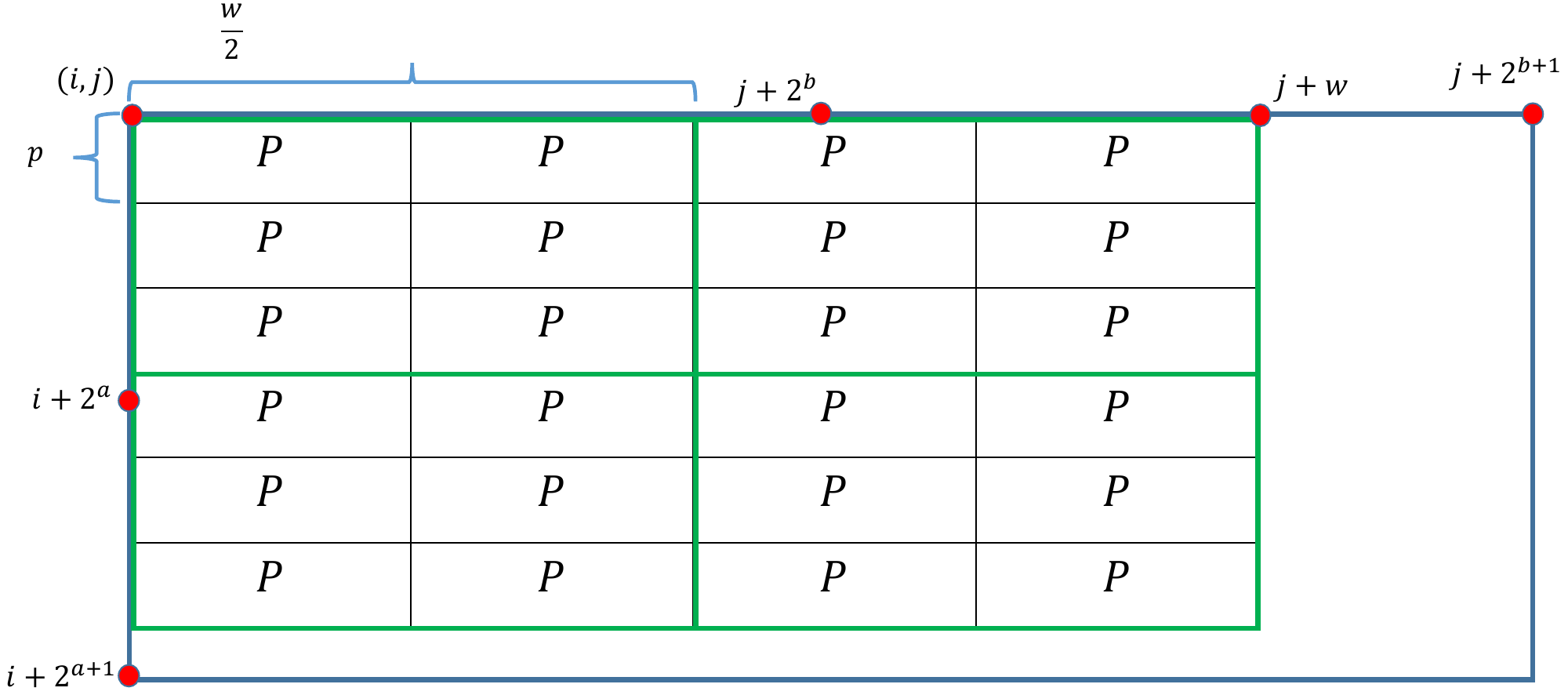}
\caption{$A[i \dd i + 2^{a+1})[j \dd j + 2^{b+1})$ is the blue border rectangle and a candidate quartic of maximal height, i.e., $P^{6,4}$, is shown in green.}\label{fig:thin-per}
\end{figure}

\subsection{Computation of Thick Quartics}
\label{ComputeThick}
\subsubsection{Computation of Supersets of $\assign(i,j)$}

In this subsection, we show how to compute, for each position $(i,j)$ of $A$, an $\cO(\log n)$-size superset of $\assign(i,j)$, in $\cO(n^2 \log n)$ time in total.

Let us fix a position $(i,j)$ of~$A$.
We say that a primitive 2D string $R$ is a \emph{skyline 2D string (for position $(i,j)$)} if and only if
\begin{itemize}
\item $R^{5,5}$ has an occurrence anchored at position $(i,j)$, and
\item There is no other primitive 2D string $P$, such that $g(P)$ dominates $g(R)$ and $P^{5,5}$ has an occurrence anchored at position $(i,j)$.
\end{itemize}
Note that a skyline 2D string $R$ does not have to be in $\assign(i,j)$.

\paragraph{Computation of Skyline 2D Strings.}
We wish to efficiently compute, for each $a \in [1 \dd \floor{\log n}]$, the widest 2D primitive string $R$, if there is any, with $\height(R) \in [2^a \dd 2^{a+1})$ and for which there is an occurrence of $R^{5,5}$ anchored at position $(i,j)$.
To this end, we start with a combinatorial result about highly periodic runs in strings $\vmeta{j}{b}$ for a fixed~$j$ and variable~$b$.

\begin{lemma}\label{lem:runs_evol}
Let $(i,j)$ be a position of $A$. Suppose that there are primitive strings $U$ and~$V$ with $|U| < |V|$ and (not necessarily distinct) integers $b_1$ and $b_2$ such that 
$U^3$ occurs at position $i$ of $\vmeta{j}{b_1}$ and $V^3$ occurs at position $i$ of $\vmeta{j}{b_2}$.
Then, we have that $2|U|\leq |V|$.
\end{lemma}

\begin{proof}
Towards a contradiction, suppose that $|V|<2|U|$.
Let $b:=\min\{b_1,b_2\}$ and $S:=\vmeta{j}{b}$.

\emph{Case I}: $b_1 \leq b_2$.
By our assumptions, $|V|+|U| < 3|U|$.
Then, as $|V|$ is not a multiple of $|U|$, $S[i+|V| \dd i+|V|+|U|)$ must equal some non-trivial rotation of $U$ (defined in  Section~\ref{preliminaries}).
On the other hand, $S[i+|V| \dd i+|V|+|U|)=S[i \dd i+|U|)=U$ as $|V|$ is a period of $S[i \dd i+3|V|)$.
We obtain a contradiction, as the primitive string $U$ cannot match any of its non-trivial rotations.

\emph{Case II}: $b_1 > b_2$.
The period $p$ of $S[i \dd i+3|U|)$ divides $|U|$.
Our assumptions imply that $|V|+p < 3|U| < 3|V|$.
Further, $S[i+|V| \dd i+|V|+p)=S[i \dd i+p)$ as $|V|$ is a period of $S[i \dd i+3|V|)$.
Since $S[i \dd i+p)$ is primitive, it cannot match any of its non-trivial rotations, and hence $S[i \dd i+|V|)=S[i \dd i+p)^k$ for some $k$. This contradicts the primitivity of $V$ and thus concludes the proof.
\end{proof}

The above lemma implies that there is at most one candidate for the height of $R$.
We next show how to efficiently compute this candidate or conclude that none exists.
We compute as a batch of $\cO(n \log n)$ candidates for a fixed $j$ and all the $\cO(n\log n)$ choices of $i$ and $a$.

\begin{lemma}\label{lem:vert_run}
Let $j \in [1 \dd n]$.
In $\cO(n \log n)$ time in total, we can compute, for each $i \in [1 \dd n]$, for each $a \in [1 \dd \floor{\log n}]$, the unique integer $h$ with $\floor{\log h}=a$, for which there exists some $b \in [1 \dd \floor{\log n}]$ such that a string $U^5$, where $U$ is a primitive string of length $h$, occurs at position $i-2h$ of~$\vmeta{j}{b}$. The subinterval $I_h$ of $[1 \dd \floor{\log n}]$ that contains exactly those integers $b$ for which this holds is also returned alongside $h$.
\end{lemma}

\begin{proof}
For each $b>1$, for each run $D=\vmeta{j}{b}[x \dd y]$ with period $p$, there is a unique run $\pi(D)$ in string $\vmeta{j}{b-1}$ that contains $\vmeta{j}{b-1}[x \dd y]$ and whose period divides $p$.
It can be computed in $\cO(1)$ time using a periodic extension query in $\vmeta{j}{b-1}$.

We build a forest $\mathcal{F}$ whose nodes are runs $D$ with $|D|\geq 5\cdot \per(D)$ of $\vmeta{j}{b}$, for $b \in [1 \dd \floor{\log n}]$.
Each run of $\vmeta{j}{1}$ is a root, while the parent of any other run $D$ is $\pi(D)$.
We preprocess $\mathcal{F}$ in $\cO(n\log n)$ time so that each run/node in $\mathcal{F}$ stores its depth as well as a \emph{special pointer} to its lowest ancestor that has a different period.

We process the runs that are represented in this forest in a left-to-right manner.
That is, we consider positions $i=1,2,\ldots,n$ in this order while maintaining the information about runs in all
strings $\vmeta{j}{b}$, for $b \in [1 \dd \floor{\log n}]$, simultaneously.
A run~$D$ is activated $2\cdot\per(D)$ positions after it starts and deactivated $3\cdot \per(D)-1$ positions before it ends.
When processing a position $i \in [1\dd n]$, we first deactivate runs as necessary (bottom-up) and then activate runs as necessary (top-down), maintaining the leaves of the subforest induced by the active nodes.
This guarantees that all ancestors of an active node are also active.
Then, we traverse the subforest induced by the active nodes by following special pointers, starting from the leaves.
The number of traversed special pointers is $\cO(\log n)$ due to \cref{lem:runs_evol}: for any $k\in [1 \dd \floor{\log n}]$,
all active runs with periods from $[2^{k}\dd 2^{k+1})$ have the same period $p$, thus forming a single path
in the subforest, and the traversal can be implemented in time proportional to the number of such paths.
Each traversed special pointer gives us a period $p$ and a subinterval $I_h \subseteq [1 \dd \floor{\log n}]$ containing all $b$
such that there is an active run $D$ at position $i$ in $\vmeta{j}{b}$ with $\per(D)=p$.
The value of $p$ is stored at the tail of the special pointer, and $I_{h}$ consists of the depths of all nodes between
its tail (inclusive) and head (exclusive).
\end{proof}

\subparagraph{Processing the candidate vertical period.}
Let us now fix a position $(i,j)$ and an $a \in [1 \dd \floor{\log n}]$ for which \cref{lem:vert_run} returns a candidate $h$, accompanied by a subinterval $I_h$ of $[1 \dd \floor{\log n}]$.
Now, we need to compute the widest primitive 2D string $R$ with $\height(R)=h$ such that $R^{5,5}$ has an occurrence anchored at position $(i,j)$ if one exists.

Let $\ell=\floor{\log (5h)}$.
Recall that $\mathcal{H}$ is the collection of all strings $\hmeta{i}{a}$.
There is a subset $\mathcal{C}$ of $\mathcal{H}$ of size at most three, such that the union of the fragments of $A$ from which the metastrings in $\C$ originate equals $A[i-2h \dd i+3h)[1 \dd n]$, and each of these fragments overlaps with the next one (in the order induced by the indices of their topmost rows) by at least $h$ rows.
We can assume that $\mathcal{C}$ contains $T:=\hmeta{i-2h}{\ell}$ and the elements of some subset of $\{\hmeta{i+2^{\ell}-h}{\ell}, \hmeta{i+3h-1-2^{\ell}}{\ell}\}$.

In a preprocessing step, we apply the following lemma with $k=2$ to each of the strings in $\mathcal{H}$ in $\cO(n^2 \log n)$ time in total.
\begin{lemma}\label{lem:bit_runs}
Given a string $S$ of length $n$ and a non-negative integer $k$, we can compute, in $\cO(n)$ time, for each $j \in [1 \dd n]$, a bitvector $\beta(S,j)$ of size $\floor{\log n}$, such that the $b$-th bit of $\beta(S,j)$ is set if and only if there exists a primitive string $U$ such that $S[j-k|U| \dd j+3|U|)=U^{k+3}$ and $\floor{\log |U|}=b$.
\end{lemma}
\begin{proof}
We perform a line-sweeping algorithm on $S$ that processes all runs whose length is at least $k+3$ times larger than their period.
These runs can be computed in $\cO(n)$ time~\cite{KK:99,runstheorem}.
Now, let us start by setting $\beta(S,j)$ to be an all-zeroes bitvector of size $\floor{\log n}$.
As we sweep over $S$, in a left-to-right manner, for each of the specified runs $D=S[x \dd y]$, we set the $\floor{\log(\per(D))}$-th bit of the maintained bitvector, i.e., $\beta(S,j)$, when position $x+k\cdot\per(D)$ in $S$ is processed and unset it when position $y-3\cdot \per(D)+1$ in $S$ is processed.
For each $j$, we set $\beta(S,j)$ to be equal to the maintained bitvector just after position $j$ in $S$ is processed.
The correctness of this approach follows by \cref{lem:runs_evol}, which implies that, for each position $j$ in $S$, for each $a \in [1 \dd \floor{\log n}]$ there is at most one run $D$ at a time such that $S[j \dd j+3\cdot\per(D))$ is contained in $D$ and $\floor{\log(\per(D))}=a$.
\end{proof}

\begin{observation}\label{obs:bit}
If we have an occurrence of a 2D string $R^{5,5}$ for a primitive 2D string $R$ of height $h$ anchored at position $(i,j)$, then, for each $X \in \C$, the $\floor{\log(\width(R))}$-th bit of $\beta(X,j)$ is $1$.
\end{observation}

To facilitate the efficient computation of the sought width, we need a converse version of the above statement.

\begin{lemma}\label{lem:bit}
Suppose that the $b$-th bit is set in $\beta(X,j)$ for all $X \in \C$.
Let $U$ be a primitive string such that $T[j-2|U| \dd j+3|U|)=U^5$ and $\floor{\log |U|}=b$.
Then, if the vertical period of $A[i-2h \dd i+3h)[j \dd j+3|U|)$ is $h$, $R:=A[i-2h \dd i-h)[j \dd j+|U|)$ is a primitive 2D string and $R^{5,5}$ has an occurrence anchored at position $(i,j)$.
\end{lemma}
\begin{proof}
Let $w:=|U|$.
By the assumptions of the lemma and the fact that $\height(T)> 2h$, the horizontal and vertical periods of $A[i-2h \dd i-2h+\height(T))[j-2w \dd j+3w)$ are $w$ and~$h$, respectively.
The primitivity of $R$ is immediate.

Now, note that if $\C=\{T\}$ we are already done as this can only be the case if $5h=2^\ell$. Henceforth, we can thus suppose that this is not the case.
Consider $Y\in\C$ such that the fragments of $A$ from which $T$ and $Y$ were built (as metastrings) overlap by at least $h$ rows, and let $O$ be this overlap.
It readily follows that the horizontal period of $O[1\dd h)[j-2w \dd j+3w)$ is $w$.
Further, by an argument analogous to that in the proof of \cref{lem:runs_evol}, $w$ must coincide with the length of the primitive string $V$ such that $V^5$ has an occurrence at position $i-2|V|$ of $Y$ and satisfies $\floor{\log |V|}=b$.
If $\C$ also contains a third element $X$, we also apply the same reasoning to $X$ and $Y$.
In either case, we conclude that $A[i-2h \dd i+3h)[j-2w \dd j+3w)$ has horizontal period $w$. Since the vertical period of $A[i-2h \dd i+3h)[j \dd j+3w)$ is $h$, we have an occurrence of $R^{5,5}$ anchored at position $(i,j)$. This completes the proof of the lemma.
\end{proof}

Now, our first goal is to compute a set $B$ that contains three candidates for $\floor{\log (\width(R))}$.
Initialize $B$ as the empty set.
Let $I=[b_1 \dd b_2]$ be the interval returned by the application of \cref{lem:vert_run} together with candidate $h$.
Let $\beta^*$ be the bitvector obtained by an AND operation performed on the bitvectors $\beta(X,j)$ for all $X \in \C$.
By \cref{obs:bit}, we only need to consider indices of bits of~$\beta^*$ that are set.
If the $b$-th bit of $\beta^*$ is set and $b \in [b_1 \dd b_2)$ then the requirements of \cref{lem:bit} are satisfied.
All such indices $b$ can be computed by applying a bitmask whose $x$-th bit is set if and only if $x \in [b_1 \dd b_2)$.
Our first candidate is the index $b^*$ of the least significant set bit of the obtained vector $\beta^*$, if there is one; we insert it to $B$.
Further, observe that we cannot have $\floor{\log (\width(R)))} \in [1 \dd b_1-1) \cup (b_2 \dd \floor{\log n}]$ since the vertical period of $A[i-2h \dd i-h)[j \dd j+\width(R))$ is not $h$.
Finally, insert to $B$ each $y\in \{b_1-1,b_2\}$ if the $y$-th bit of $\beta^*$ is set.

If $B$ is empty there is nothing to be done as there is no primitive 2D string $R$ to be returned.
We thus assume that $B$ is non-empty.
For each element $x$ of $B$, we compute the length $w_x$ of the primitive string $U$ such that $U^5$ occurs at position $i-2|U|$ of $T$ and $\floor{\log |U|}=x$ as follows.
For each $x \in B$, we query for the primitive squares that occur at position $j$ and have length in $[2^x \dd 2^{x+1})$ using~\ref{lem:long_sq}.
For each result of these queries, i.e., for each of the at most two returned primitive squares $UU$ for each $x \in B$, we perform a periodic extension query that takes $\cO(1)$ time to check if this square extends to an occurrence of $U^5$ at position $i-2|U|$ in $T$; this will give us a unique $U$ due to \cref{lem:runs_evol}, yielding $w_x=|U|$.
Next, for each $x \in B$, we test whether the vertical period of $A[i-2h \dd i+3h)[j \dd j+3w_x)$ is $h$ using 2-period queries for strings $\vmeta{j}{y}[i-2h \dd i+3h)$ and $\vmeta{j+3w_x-1-2^{y}}{y}[i-2h \dd i+3h)$, where $y=\floor{\log (3w_x)}$, and checking whether both returned periods are equal to $h$.\footnote{The check for $b^*$, if it exists, is guaranteed to be successful.}
Then, $\width(R)$ is equal to the largest $w_x$ for which this test is successful, if there is one, due to \cref{lem:bit}.

Finally, we postprocess the obtained $\cO(\log n)$ elements in $\cO(\log n)$ time to remove any elements that are dominated, thus obtaining the sought skyline of primitive 2D strings.

\paragraph{Computation of Dominated 2D Strings.}

\renewcommand{\R}{\mathcal{R}}

Let the skyline 2D strings be $S_1, \ldots, S_\ell$, sorted in the order of decreasing width (and thus increasing height); the set of these 2D strings is henceforth denoted by $\S$.

We say that a 2D string $U$ (strictly) dominates a 2D string $V$ if and only if $g(U)$ (strictly) dominates $g(V)$.
We also refer to chains and antichains of 2D strings based on this dominance relation.
By \cref{lem:onechain}, for each element of $\S$, the primitive 2D strings that span it horizontally (resp.~vertically) form a chain.

For a set of strings $\mathcal{X}$, let us denote $g(\mathcal{X}):=\{g(X) :X\in \mathcal{X}\}$.
Our aim is to compute an $\cO(\log n)$-size set $\mathcal{G} \subset [1\dd \log n]^2$ such that $\mathcal{G} \supseteq g(\assign(i,j))$.
We initialise $\mathcal{G}$ as $\mathcal{S} \cup \{(a,b): (a,b) \text{ is dominated by } (a',b) \in \mathcal{S} \text{ or } (a,b') \in \mathcal{S}\}$.
Let $\R$ be the set of primitive 2D strings that are dominated by at least one element of $\S$ and span either horizontally or vertically each of the elements of $\S$ that dominate them.
Our main aim is to compute $g(\R):=\{g(R) :R\in \R\}$ and then set $\mathcal{G} = \mathcal{G} \cup g(\R)$, which by \cref{lem:span} is a superset of $g(\assign(i,j))$.
In the end, for each $(a,b) \in \mathcal{G}$ we infer a primitive 2D string $R$ with $g(R)=(a,b)$ and check whether $R^{5,5}$ has an occurrence anchored at position $(i,j)$ using a constant number of $2$-period queries for fragments of strings from $\mathcal{H} \cup \mathcal{V}$.
We next show how to compute $g(R)$.

\begin{fact}
For any $k$, a 2D string that horizontally (resp.~vertically) spans $S_k \in \S$ also spans all $S_t$ with $t > k$ (resp.~$t<k$) that dominate it.
\end{fact}
Consider a primitive 2D string $R \in \R$ and let the elements of $\S$ that dominate it be $S_k, \ldots, S_t$.
Then, $R$ horizontally spans a (possibly empty) suffix of $S_k, \ldots, S_t$, $R$ vertically spans a (possibly empty) prefix of $S_k, \ldots, S_t$, and the sum of the lengths of the suffix and the prefix is at least $t-k+1$.
In what follows, we will process the elements of $\S$ in the order of decreasing width.
We will aim to compute $g(R)$ for each $R \in \R$:

\begin{itemize}
\item When processing the widest element of $\S$ that strictly dominates $R$ (and is spanned horizontally by it),
\item[] or, if such a point does not exist,
\item when processing the thinnest element of $\S$ that strictly dominates $R$ (and is spanned vertically by it).
\end{itemize}

We first process $S_1$, reporting (the dimensions of) the following 2D strings:
\begin{itemize}
\item The primitive 2D strings that vertically span $S_1$ and are not dominated by any other element of $\mathcal{S}$.
\item The primitive 2D strings that horizontally span $S_1$.
\end{itemize}
When processing $S_k$, for $k \in (1 \dd \ell)$, we report (the dimensions of) the following 2D strings:
\begin{itemize}
\item The primitive 2D strings that vertically span $S_k$ and are not dominated by $S_{k+1}$.
\item The primitive 2D strings that horizontally span $S_k$ and are not dominated by $S_{k-1}$.
\item The primitive 2D strings that horizontally span $S_k$ and vertically span $S_{k-1}$. We mark each 2D string that is reported by this procedure. We break the procedure whenever we encounter a 2D string that has already been marked.
\end{itemize}
Finally, we treat $S_{\ell}$ separately in $\cO(\log n)$ time by computing (the dimensions of) all primitive 2D strings that span it either horizontally or vertically.

The above procedure reports $\cO(\log n)$ (dimensions of) primitive 2D strings due to \cref{lem:geom}.
As we will verify each of them separately in the end, we only need to show that we do not miss any of the 2D strings.
By~\cref{cor:twochains}, missing any of them could only be due to the imposed stopping condition, that is, breaking upon encountering a marked 2D string.
To that end, we utilize the following simple lemma, which,
intuitively, states that horizontal and vertical spanning are transitive properties.

\begin{lemma}\label{lem:transitive}
Let $X$, $Y$, and $Z$ be primitive 2D strings that satisfy $g(X) \leq g(Y) \leq g(Z)$.
Then, the following two statements are equivalent:
\begin{enumerate}[label = ({\alph*})]
\item $X$ horizontally (resp.~vertically) spans $Y$ and $Y$ horizontally (resp.~vertically) spans $Z$;\label{it:tra}
\item $X$ and $Y$ horizontally (resp.~vertically) span $Z$.\label{it:trb}
\end{enumerate}
\end{lemma}
\begin{proof}
By symmetry, it suffices to prove the statement for horizontal spanning.
Recall that \cref{fact:span} states that if $S$ and $T$ are two primitive 2D strings such that $S$ spans $T$ horizontally, then the horizontal period of $\hspan{S}{T}$ is $\width(S)$.
Since $\width(T)$ is also a horizontal period of $\hspan{S}{T}$, $\width(S)$ divides $\width(T)$ and hence $\hspan{S}{T}=S^{3,2s}$ for some integer $s>1$.

\ref{it:tra} $\Rightarrow$ \ref{it:trb}:
We have to show that $X$ horizontally spans $Z$.
Since
$\hspan{Y}{Z}=Y^{3,2y}$ for some integer $y\geq 1$
and
$\hspan{X}{Y}=X^{3,2x}$ for some integer $x\geq 1$,
we have that $\hspan{X}{Z}=X^{3,2xy}$.

\ref{it:trb} $\Rightarrow$ \ref{it:tra}:
We have to show that $X$ horizontally spans $Y$.
We have that
$\hspan{Y}{Z}=Y^{3,2y}$ for some integer $y\geq 1$
and
$\hspan{X}{Z}=X^{3,2x}$ for some integer $x\geq 1$.
Further, since $\width(X)\leq \width(Y)$, we have that $x\geq y$.
First, observe that by the periodicity lemma (\cref{lem:FW}), $3\cdot \height(X) < 2\cdot \height(Y)$, else $Y$ would not be primitive.
Then, $\hspan{X}{Y}$,  which consists of the first $2\cdot \width(Y)$ columns of $X^{3,2x}$  has both $\width(X)$ and $\width(Y)$ as its horizontal periods.
By combining another application of the periodicity lemma with the fact that $X$ is primitive, $\width(X)$ must divide $\width(Y)$.
This implies that $\hspan{X}{Y}=X^{3,2x/y}$, concluding the proof.
\end{proof}

Due to the above lemma, for any two elements $X$ and $Y$ in the chain of primitive 2D strings, where $g(X)\leq g(Y)$, and for some $k$,
both $X$ and $Y$ horizontally span $S_k$ and vertically span $S_{k-1}$, we have that $X$ spans $Y$ both vertically and horizontally.
Conversely, each 2D string $X'$ that spans~$Y$ both horizontally and vertically, must also horizontally span $S_k$ and vertically span $S_{k-1}$.
Thus, the intersection of:
\begin{itemize}
\item a chain of primitive 2D strings that horizontally span $S_k$ and vertically span $S_{k-1}$, and
\item a chain of primitive 2D strings that horizontally span $S_{k'}$ and vertically span $S_{k'-1}$
\end{itemize}
for $k\neq k'$ consists of the longest common suffix of these chains (in decreasing order with respect to the mapping $g(\cdot)$).
Thus, we do not lose any primitive 2D strings due to our stopping condition, which enables us to avoid reporting the same 2D string multiple times (e.g., a primitive 2D string that spans all elements of $\S$ both horizontally and vertically).

We have thus reduced the problem in scope, in $\cO(n^2 \log n)$ time over all positions, to the problem of being able to generate chains and intersections of chains online in the natural order from the bottom-right in $\cO(1)$ worst-case time per chain element.

\subparagraph{Generation of chains.}
It suffices to discuss how to generate the chain of primitive 2D strings that horizontally span a skyline 2D string $S$.
Let $(\alpha, \beta) := g(S)$.
We iterate over all primitive strings~$U$ such that $U^3$ occurs at position $i$ of $\vmeta{j}{\beta}$ and $|U|\leq 2^\alpha$.
Consider one such $U$ and let $p:=\floor{\log |U|}$.
We compute the period $q$ of $\hmeta{i}{p+1}[j \dd j+2\cdot \width(S))$, which is at most $\width(S)$ using a $2$-period query and conclude that $A[i \dd i+|U|)[j \dd j+q)$ horizontally spans~$S$.
Thus, we return $(p,\floor{\log q})$.
Hence, processing a single value of $\floor{\log|U|}$ takes $\cO(1)$ time and is guaranteed to return a primitive 2D string that horizontally spans $S$.
We efficiently retrieve all values of $p$ one by one using the output of \cref{lem:bit_runs} applied to $\vmeta{j}{\beta}$ with $k=0$.
This lemma returns a bitvector $v$ whose $a$-th bit is set if and only if there is a sought~$U$ that satisfies $\floor{\log |U|}=a$.
It thus suffices to iterate over the set bits of this bitvector that are lower than $\alpha$.

\subparagraph{Generation of intersections of chains.}
We use tabulation, exploiting the following lemma.

\begin{lemma}\label{lem:four_russians}
There exists a primitive 2D string $R$ with $g(R)=(a,b)$ that spans $S_k$ horizontally and $S_{k-1}$ vertically if and only if the following conditions hold, where $(a_1,b_1)=g(S_k)$ and $(a_2,b_2)=g(S_{k-1})$:
\begin{itemize}
\item There is a primitive string $W$ of length $w \in [2^b \dd 2^{b+1})$ such that $W^3$ has an occurrence at position $j$ of $\hmeta{i}{a_2+1}$
and the period of $\hmeta{i}{a+1}[j \dd j+|S_{k}|)$ is in $[2^b \dd 2^{b+1})$;
\item There is a primitive string $U$ of length $h \in [2^a \dd 2^{a+1})$ such that $U^3$ has an occurrence at position $i$ of $\vmeta{j}{b_1+1}$,
and the period of $\vmeta{j}{b+1}[i \dd i+|S_{k-1}|)$ is in $[2^a \dd 2^{a+1})$.
\end{itemize}
\end{lemma}
\begin{proof}
The ($\Rightarrow$) direction is immediate.

($\Leftarrow$):
The specified occurrence of $W^3$ implies that the first $3w$ columns of $S_{k-1}$ have period $w$.
This means that there is a primitive 2D string $Y$ of width $w$ that vertically spans~$S_{k-1}$.
In addition, the assumption that the period of $\vmeta{j}{b+1}[i \dd i+|S_{k-1}|)$ is in $[2^a \dd 2^{a+1})$ implies that so is the height of $Y$.
Then, due to \cref{lem:runs_evol}, the specified occurrence of $U^3$ implies that $\height(Y)=h$.
Hence, the string $A[i \dd i+h)[j \dd j+w)$ vertically spans $S_{k-1}$.
By symmetry, the same string horizontally spans $S_{k}$.
\end{proof}

The above lemma means that, given
\begin{enumerate}[label = {(\arabic*)}]
\item The floor of the logarithm of the period of $\vmeta{j}{b}[i \dd i+|S_{k-1}|)$ for all $b$,\label{vert}
\item The floor of the logarithm of the period of $\hmeta{i}{a}[j \dd j+|S_{k}|)$ for all $a$, and\label{hor}
\item The bitvector computed by an application of \cref{lem:bit_runs} for $k=0$ for position $j$ of $\hmeta{i}{a_2+1}$ and position $i$ of $\vmeta{j}{b_1+1}$,\label{easy}
\end{enumerate}
one can compute the set $\{g(R): R \text{ spans $S_k$ horizontally and $S_{k-1}$ vertically}\}$ in $\cO(\log^2 n)$ time by iterating over all possible pairs $(a,b)$.
We next show that the information specified in \ref{vert} and \ref{hor} can be represented in $\cO(\log n)$ bits and can be computed efficiently (in a batched computation), as is already the case for \ref{easy}.
This will allow us to use tabulation.
In what follows, we focus on the representation and the computation of \ref{vert} as \ref{hor} is symmetric.

We encode the information on how the floor of the logarithm of the period of  $\vmeta{j}{b}[i \dd i+|S_{k-1}|)$ changes
when we iterate over $b\in [0\dd \floor{\log (\height(S_{k-1}))}]$. After increasing $b$, this value either stays
the same or increases. As the logarithm of the period always belongs to $[0\dd \floor{\log (\height(S_{k-1}))}]$,
this suggests the following natural encoding. We write down a sequence $0^{x_{0}}1 0^{x_{1}}1 0^{x_{2}}\ldots 0^{x_{\ell}-1} 1$,
where $\ell = \floor{\log (\height(S_{k-1}))}$. The $b^{\text{th}}$ 1 (counting from 0), corresponds to 
$\vmeta{j}{b}[i \dd i+|S_{k-1}|)$, and the total number of 0s before it is equal to the floor of the logarithm of its period.
This represents the required information and consists of $\cO(\log n)$ bits.

Further, the specified representation can be computed, for a fixed $j$, over all choices of $i$ and~$k$ in $\cO(n \log n)$ total time.
This can be done by processing all runs in $\vmeta{j}{b}$ in a manner analogous to that in the proof of \cref{lem:vert_run}.
We build a forest of all runs whose length is at least three times larger than their period (instead of five), and then, for each position, for each of the $\cO(\log n)$ considered root-to-leaf paths, we build the specified representation in time linear in the number of special pointers on that path.
The result follows by the fact that the total number of special pointers traversed for each position is $\cO(\log n)$.
We need to retrieve the bitvector constructed for a leaf and apply a bitmask to keep only the bits that correspond to runs whose length is at least $|S_{k-1}|$ and period is at most $|S_{k-1}|$.

A direct attempt at using tabulation would be to preprocess the answer for every combination of the bitvectors
encoding the information described in \ref{vert}, \ref{hor}, and \ref{easy}. This is however $c\cdot \log n$ bits, for $c=6$, so we cannot
simply precompute the answer for each such input.
Instead, we apply the standard approach of partitioning each of the bitvector into a constant number of smaller bitvectors,
so that we have smaller instances of the same problem for which we can naively precompute the answers.
This is described in more detail below.

Let $\alpha>0$ be a constant to be chosen later.
The bitvectors are partitioned into fragments corresponding to a range of at most $\alpha \cdot \log n$ values of $a$ and $b$, respectively.
This partition is chosen to additionally ensure that the floor of the logarithm increases by at most $\alpha \cdot \log n$ when we iterate over
those values. Thus, the whole description of what is the floor of the logarithm consists of $2\alpha\cdot \log n$ bits,
the initial value of $a$ or $b$, and the floor of the logarithm of the period for this initial value, which takes
$2\log\log n$ additional bits, and we only need to create $2/\alpha$ fragments from each bitvector.
Then, we precompute the elements of the sought chain for every pair of possible fragments. By adjusting the sufficiently
small constant $\alpha>0$ this takes $o(n)$ time.
For each possible bitvector (describing the situation for the values of $a$ or $b$), we preprocess its partition
into $2/\alpha=\cO(1)$ fragments. This takes only $\cO(n^{2}\log n)$ time and space, as the bitvectors are of length $2\log n$.
Then, when given the bitvectors, we first retrieve their partition into $2/\alpha=\cO(1)$ fragments each.
Next, we iterate over all $(2/\alpha)^{2}=\cO(1)$ combinations of fragments describing a range of values of $a$ and $b$,
retrieve the precomputed answer, and return the union of their precomputed chains.
The chains are stored in the natural order, and we return $(2/\alpha)^{2}=\cO(1)$ pointers to their lists,
which allows us to stop as soon as we see an element that has been already reported while spending only
constant time per reported element.

\subparagraph*{Verification.}
Now, for each pair $(a,b)$ in $\mathcal{G}$, we aim to compute the primitive 2D string $R$, if one exists, such that $R^{5,5}$ has an occurrence anchored at position $(i,j)$.
We compute the height of $h$, if it exists, as follows.
\begin{itemize}
\item Case I: $2h> 3\cdot 2^a$. We have that $3h\geq  4\cdot 2^{a}$.
Then, since $2h< 2\cdot 2^{a+1}$, $h$ must be the period of $\vmeta{j}{b+1}[i \dd i+2^{a+2})$, which can be computed in $\cO(1)$ time using a 2-period query.
\item Case II: $2h \leq 3\cdot 2^a$. $h$ is the period of $\vmeta{j}{b+1}[i \dd i+3\cdot 2^a)$, which can be computed in $\cO(1)$ time using a 2-period query.
\end{itemize}
We perform both of the above $2$-period queries and try to extend the obtained periods using periodic extension queries in $\vmeta{j}{b+1}$; at most one of these extensions will be successful due to \cref{lem:runs_evol}.
We symmetrically compute a candidate width $w$ or conclude that no such $R$ exists.
Finally, assuming that we have obtained a candidate height $h$ and a candidate width $w$, we compute $R$ or conclude that none exists by checking, using a constant number of 2-period queries on strings from $\mathcal{H} \cup \mathcal{V}$, whether $A[i-2h \dd i+3h)[j-2w \dd j+3w)$ has vertical period $h$ and horizontal period $w$.
\subsubsection{Computation of Thick Quartics from Sets $\assign(i,j)$}
We process all found occurrences of $R^{5,5}$ as in \cite[Section 5]{CRRWZ}. While the complexity claimed
in previous work was $\cO(n^{2}\log^2 n)$, the underlying method can be used to process a set $S$ of found
occurrences of $R^{5,5}$ in $\cO(|S|)$ time to calculate, for every $x\geq 5$, the largest $y\geq 5$
such that an occurrence of $R^{x,y}$ can be composed of the given occurrences,
assuming that the pairs in the set are lexicographically sorted.
Overall, this sums up to $\cO(n^{2}\log n)$.
Note that while sorting a single set with radix sort would take $\cO(|S|+n)$ time, we can sort all the sets
together, paying the additive $\cO(n)$ only once.
We briefly explain how to process a set $S$ in the claimed time complexity. The reader can find a detailed
description in \cite[Section 5]{CRRWZ}. We stress that this is not meant to be a new result, but just
a summary of a known method given for completeness.

First, we partition the pairs $(i,j)$ into subsets corresponding to distinct values of $(i\bmod \height(R),j\bmod \width(R))$. Any fragment
of the form $R^{x,y}$ with $x,y\geq 5$ is composed of occurrences of $R^{5,5}$ from the same subset.
Thus, each subset can be processed separately. Further, we note that partitioning into subsets can be done
in $\cO(|S|)$ time using a $\height(R)\times \width(R)$ array of initially empty lists. The array is allocated only
once at the very beginning of the whole procedure and reused for every $S$ (after having processed each $S$, we clear the array).
Further, we can assume that the pairs in every subset are sorted.

We explain how to process a subset $S'\subseteq S$. The problem now reduces to the following.
We are given a set of white cells in a grid $[1 \dd n]^{2}$ and want to report all pairs $(h,w)$ such that
there exists an $h\times w$ rectangle consisting of white cells. Further, the white cells are lexicographically
sorted. This allows us to process them row-by-row while maintaining an array $H[1\dd n]$. When processing
row $r$, we want every $H[c]$ to be equal to the height of the tower of white cells above position $(r,c)$,
that is, the largest $k$ such that $(r,c),(r,c-1),\ldots,(r,c-k+1)$ are white cells. This can be maintained
in time proportional to the number of white cells in the previous and the current row.
Next, we process maximal fragments $H[i\dd j]$ consisting of non-zero entries. The goal is to process such
a fragment in $\cO(j-i+1)$ time. We first find, for every $k\in [i\dd j]$, the nearest value smaller than $H[k]$ to the left
and to the right of $k$ in $H[i\dd j]$, denoted by $H[\textsf{left}[k]]$ and $H[\textsf{right}[k]]$, respectively.
This can be done with two sweeps over $[i \dd j]$ (left-to-right and right-to-left)
while maintaining a stack in $\cO(j-i+1)$ time. Then, for $x=H[k]$, we set $y=\max\{y, \textsf{right}[k]-\textsf{left}[k]-1\}$.

\bibliographystyle{plainurl}
\bibliography{biblio}
\end{document}